\documentclass[11pt, a4paper]{article}

\usepackage[margin=1in]{geometry} 
\usepackage{amssymb}     
\usepackage{graphicx}  
\usepackage[colorlinks=true, linkcolor=black]{hyperref}
\usepackage{lineno}
\usepackage{authblk}          
\usepackage{hyperref}             
\usepackage{abstract}  
\usepackage{amsmath}
\numberwithin{equation}{subsection}
\usepackage[amsthm]{ntheorem} 

\theoremstyle{plain}
\theoremheaderfont{\bfseries}
\theorembodyfont{\itshape}
\newtheorem{theorem}{Theorem}
\hypersetup{
    colorlinks=true,
    citecolor=blue,               
    linkcolor=black,              
    urlcolor=blue                 
}
\theoremstyle{nonumberplain} 
\theoremheaderfont{\bfseries}
\theorembodyfont{\normalfont}
\newtheorem{remark}{Remark:} 

\title{\textbf{
A Unified Discrete and Continuous Theory of Core--Halo Complexity Maximizers
}}
\providecommand{\keywords}[1]
{
  \small	
  \textbf{\textit{Keywords:}} #1
}

\author[1]{Akshat Sharma}
\affil[1]{\small \textit{Indian Institute of Science Education and Research Kolkata}} 
\affil[1]{\texttt{\textbf{Email:}as23ms078@iiserkol.ac.in}}

\date{} 

\begin{document}

\maketitle

\noindent \begin{abstract}

The maximization of statistical complexity has long been associated with the emergence of probability distributions lying between perfect order and complete disorder. While previous studies have shown that complexity-maximizing distributions exhibit a two-level structure in finite discrete systems, an analogous unified treatment for both discrete and continuous probability spaces has remained unavailable. In this work, we develop a general variational framework for a generalized statistical complexity constructed from Shannon and R\'{e}nyi entropies. We derive a common stationary equation governing both discrete probability masses and continuous probability densities and prove that every stationary solution necessarily possesses exactly two probability levels, establishing a universal core–halo structure. We further demonstrate that the optimization problem reduces to a single multiplicity parameter and prove that the global complexity maximum is attained by the smallest admissible core, corresponding to a single dominant state in the discrete case and an infinitesimal core in the continuous limit. These results provide a complete analytical characterization of the complexity-maximizing distributions and reveal a common mathematical structure underlying complexity optimization in both discrete and continuous settings. The framework establishes a unified foundation for generalized statistical complexity with potential applications in statistical mechanics, information theory, and the analysis of complex systems.

\end{abstract}
\vspace{2em}
\vspace{2em} 

\tableofcontents
\newpage
\section{Introduction}
\label{introduction}
The notion of complexity occupies a unique position in modern science and has emerged as one of the most challenging and interdisciplinary problems connecting physics, chemistry, biology, information theory, machine learning, and the broader science of complex systems. While the concepts of order and disorder have been successfully formalized through the frameworks of mechanics, thermodynamics, and information theory, the mathematical characterization of complexity itself remains elusive. Complexity appears to arise in an intermediate regime between complete order and complete randomness, where structured organization emerges through the interplay of competing tendencies toward regularity and disorder.

The importance of this problem has been recognized for many decades. Schr\"{o}dinger, in his celebrated monograph \textit{What is Life?}\cite{schrodinger1944what}, emphasized that living systems maintain organized structures while continuously exchanging energy and information with their environment. Wiener viewed organization and communication as central ingredients underlying cybernetic systems\cite{wiener1948cybernetics}, whereas Prigogine highlighted the role of nonequilibrium processes and dissipative structures\cite{prigogine1978time} in generating complexity. Haken's theory of synergetics\cite{haken1977synergetics} and the subsequent development of complex systems science further reinforced the idea that complexity represents a distinct regime characterized neither by perfect order nor by complete chaos. Examples of such behavior are abundant in nature: atomic and molecular systems exhibit intricate electronic structures, biological organisms maintain highly organized networks of interacting components, ecological communities display rich patterns of cooperation and competition, and neural networks give rise to emergent intelligence. Even the large-scale structure of the universe exhibits a remarkable coexistence of order and randomness. These observations suggest that complexity is fundamentally associated with the appearance of organized structures and cannot be identified solely with either regularity or disorder.

A desirable measure of complexity should therefore satisfy several qualitative requirements: it should vanish in perfectly ordered systems, where a single state dominates the probability distribution, and it should also vanish in completely random systems, where all accessible states become equiprobable. Maximum complexity is thus expected to occur somewhere between these two limiting regimes, and the mathematical determination of such intermediate states remains one of the central problems in complexity science.

The earliest attempts to quantify complexity originated from algorithmic information theory through the works of Kolmogorov\cite{kolmogorov1965three}, Solomonoff\cite{solomonoff1964formal}, and Chaitin\cite{chaitin1966length}. Within this framework, the complexity of an object is identified with the length of the shortest algorithm capable of generating it\cite{lempelziv1976complexity, crutchfield1989inferring}. Although algorithmic complexity provides a mathematically rigorous notion of randomness, it possesses an important conceptual drawback: completely random sequences exhibit maximal algorithmic complexity. Consequently, algorithmic complexity measures unpredictability rather than organized structure and therefore fails to distinguish between noise and meaningful organization. Several alternative approaches have subsequently been proposed. Bennett\cite{bennett1988logical} introduced the concept of logical depth to characterize structures with a rich computational history, whereas Gell-Mann and Lloyd\cite{gellmann2004effective} developed the notion of effective complexity to separate regularities from random fluctuations. Despite their conceptual importance, these approaches remain difficult to apply to physical systems and often lack direct connections with statistical mechanics and information theory.

From the viewpoint of statistical physics, entropy naturally emerges as a measure of disorder. Shannon entropy\cite{shannon1948mathematical}, introduced in 1948, provides a quantitative measure of uncertainty and forms the foundation of modern information theory. Various generalizations, including the R\'enyi\cite{renyi1961measures} and Tsallis\cite{tsallis1988possible} entropies, have subsequently found applications throughout statistical mechanics, multifractal analysis\cite{hill1973diversity}, and quantum information theory\cite{petz1986, rajagopal1999}. Nevertheless, entropy alone cannot characterize complexity: maximal entropy corresponds to complete randomness rather than maximal organization, while minimal entropy corresponds to perfect order\cite{onicescu1966energie, beckschlogl1993thermodynamics}. Complexity therefore cannot be identified with disorder alone and necessarily requires an additional ingredient capable of capturing organization.

Among the many proposals introduced to overcome this limitation, one of the most influential is the statistical complexity proposed by L\'opez-Ruiz, Mancini, and Calbet (LMC)\cite{lopezruiz1995statistical} in 1995. Their construction rests on the idea that complexity should emerge from the competition between disorder and disequilibrium\cite{anteneodo1996features}: in its original form, LMC complexity is the product of a disorder measure (Shannon entropy) and a quantity describing the departure from equiprobability, expressed as a Euclidean distance in probability space. This framework successfully reproduces the intuitive requirement that complexity vanish at complete order and complete disorder while exhibiting nontrivial maxima in intermediate regimes\cite{calbet2001tendency, calbet2007extremum, calbet2009extremum}, and subsequent studies of its extremal distributions demonstrated that complexity-maximizing states are generally characterized by the coexistence of dominant and non-dominant probability components. More recently, explicit dynamical paths driving nonequilibrium trajectories toward these optimized configurations have been mapped out\cite{plastino2025dynamical}.

Since its introduction, the LMC complexity has been extensively applied across a wide range of systems. In atomic and molecular physics, complexity measures have been used to study electronic configurations, shell structures, and electron densities\cite{chatzisavvas2005information, sanudo2008statistical}, and comparable approaches have been extended to nuclear matter\cite{moustakidis2012statistical}, quantum systems\cite{sanudo2008complexity, lopezruiz2009generalized}, ecological networks, and financial time series. In parallel, Rosso and collaborators introduced complexity measures based on permutation entropy for the analysis of time series and dynamical systems, with applications ranging from finance and physiology to turbulence\cite{martin2006statistical, rosso2007distinguishing}. Numerous generalizations of the original LMC framework have also been proposed, replacing Shannon entropy and Euclidean disequilibrium with R\'enyi and Tsallis-based quantities\cite{catalan2002features, sanchezmoreno2014generalized, martin2000tsallis, shiner1999simple}, alongside related structural-entropy approaches to spatial localization in disordered and finite-dimensional systems\cite{pipek1992universal, varga2003renyi}.

Despite this extensive body of work, most existing studies have focused primarily on numerical applications and empirical investigations, while comparatively little attention has been devoted to the mathematical structure underlying statistical complexity itself. In particular, the determination of probability distributions that maximize complexity remains incompletely understood, and discrete and continuous probability spaces are usually treated separately, with no unified variational framework capable of describing both cases simultaneously. A further difficulty is more conceptual: the disequilibrium component of the original LMC formulation, based on Euclidean distance from equiprobability, has an appealing geometric interpretation but lacks a direct information-theoretic origin, and many generalized formulations introduce alternative disequilibrium functions in an essentially phenomenological manner.

From the perspective of information theory, however, concentration and localization are naturally quantified through information energy. Introduced by Onicescu in 1966\cite{onicescu1966energie}, information energy provides a measure complementary to entropy: whereas entropy increases with disorder, information energy increases with localization and peakedness. This duality suggests that information energy is a more natural candidate for describing disequilibrium within the framework of statistical complexity. Generalized statistical mechanics further highlights the role of escort probability distributions\cite{beckschlogl1993thermodynamics} in the study of multifractals, nonextensive systems, and generalized entropy measures; escort moments are intimately connected to R\'enyi entropies and thereby establish a direct bridge between entropy and concentration.

These observations motivate the present work. Rather than introducing generalized complexity through arbitrary modifications of the original LMC measure, we construct a framework built entirely on information-theoretic quantities. Disorder is quantified through the Hill number\cite{hill1973diversity} associated with Shannon entropy,
$$N_{1} = \exp(H_{1}),$$
which admits the interpretation of an effective number of accessible states. The concentration component is instead described through the information energy of the escort probability distribution\cite{onicescu1966energie, beckschlogl1993thermodynamics}, leading to the generalized disequilibrium
$$D_{q} = \frac{P_{2q}}{\left(P_{q}\right)^{2}}, \qquad P_{q} = \sum_{i=1}^{n}w_{i}^{q} \ \ \text{(discrete)}, \qquad P_{q} = \int\rho(x)^{q}\,dx \ \ \text{(continuous)}.$$
The generalized statistical complexity then takes the form
$$C_{q} = N_{1}D_{q} = \exp(H_{1})\,\frac{P_{2q}}{\left(P_{q}\right)^{2}}.$$
Using the connection between escort moments and R\'enyi entropies\cite{renyi1961measures}, the logarithm of the complexity can be expressed entirely in terms of entropic quantities,
$$\ln(C_{q}) = H_{1} + (2q-2)H_{q} - (2q-1)H_{2q},$$
giving rise to the variational functional
$$\Phi = H_{1} + (2q-2)H_{q} - (2q-1)H_{2q}.$$
Crucially, this functional does not arise from an arbitrary mathematical construction but follows directly from generalized statistical complexity itself, so the extremization problem investigated here possesses a clear physical and information-theoretic origin.

The principal objective of this paper is to determine the probability distributions that maximize this generalized complexity and to investigate their structural properties. We develop a unified variational approach, applicable simultaneously to discrete probability distributions and continuous probability densities, and show that although the two settings differ in the domains over which the probability variables are defined, both satisfy a common stationary equation arising from the same variational principle. A second objective is to characterize the geometric and structural properties of the corresponding extremal distributions: rather than focusing solely on numerical values of complexity, we ask whether complexity-maximizing states possess structural characteristics that are independent of the underlying probability space. Remarkably, the analysis reveals the emergence of a characteristic organization consisting of a highly concentrated core coexisting with a diffuse halo. This core--halo structure appears naturally in both discrete and continuous settings, arises directly from the competition between entropy and information energy, and is shown to be a universal feature of complexity-maximizing distributions rather than a consequence of any particular model or physical system.

The present work differs from previous studies in several important respects. First, the generalized disequilibrium employed here originates from the information energy of escort probability distributions and possesses a direct information-theoretic interpretation, rather than the phenomenological disequilibrium functions used in earlier generalizations. Second, the complexity functional is expressed entirely in terms of Shannon and R\'enyi entropies, establishing a natural connection between statistical complexity and generalized entropy theory. Third, the discrete and continuous problems are treated simultaneously within a unified variational framework rather than being investigated independently. Fourth, attention is directed toward the structure of extremal distributions themselves rather than merely toward the numerical evaluation of complexity. Finally, the analysis establishes the existence of a universal core--halo organization underlying complexity-maximizing states, suggesting that the emergence of such structures may constitute a fundamental and universal consequence of complexity optimization itself. In this sense, the present work aims not merely to introduce another generalized complexity measure, but to uncover the structural principles governing probability distributions of maximal complexity.

The remainder of this paper is organized as follows. Section 2 reviews the necessary concepts from probability theory, information theory, R\'enyi entropy, Hill numbers, escort probability distributions, and information energy. Section 3 formulates the generalized statistical complexity and derives the corresponding variational functional. Section 4 develops the unified variational framework and derives the stationary equations governing extremal distributions. Section 5 investigates the properties of the resulting solutions and establishes the emergence of core--halo structures. Section 6 addresses the continuous case and its relation to the discrete framework. Finally, Section 7 summarizes the principal results and discusses possible applications and future directions.

\section{Necessary Concepts}
\label{necessary concepts}
In order to formulate the generalized statistical complexity investigated throughout this work, it is necessary to review several concepts from probability theory and information theory. Although these quantities are individually well known, their combination gives rise to a natural variational structure whose extremal solutions constitute the principal subject of the present paper. The purpose of the present section is therefore twofold. First, we establish a common notation applicable to both discrete probability distributions and continuous probability densities. Second, we show how the generalized complexity functional investigated later in this work emerges naturally from the interplay between entropy and information energy.

\subsection{Probability distributions}

Probability theory provides the mathematical language used to describe uncertainty and randomness. Throughout the present work we shall consider both discrete probability distributions and continuous probability densities\cite{demoivre1733approximatio}.

For the discrete case, let
\begin{equation}
W=(w_1,w_2,\ldots,w_n)
\end{equation}
denote a probability distribution defined on a finite set of states. The probabilities satisfy
\begin{equation}
w_i\ge0,
\end{equation}
together with the normalization condition
\begin{equation}
\sum_{i=1}^{n}w_i=1.
\end{equation}
Thus, the set of all admissible discrete probability distributions forms the $(n-1)$-dimensional probability simplex
\begin{equation}
\Delta^{n-1}
=
\left\{
W\in\mathbb{R}^{n}:
w_i\ge0,\;
\sum_{i=1}^{n}w_i=1
\right\}.
\end{equation}

Similarly, in the continuous setting we consider a probability density function
\begin{equation}
\rho(x)\ge0
\end{equation}
satisfying
\begin{equation}
\int_{\Omega}\rho(x)\,dx=1,
\end{equation}
where $\Omega$ denotes the support of the distribution. The corresponding space of admissible continuous probability densities is
\begin{equation}
\mathcal{P}(\Omega)
=
\left\{
\rho\in L^{1}(\Omega):
\rho(x)\ge0,\;
\int_{\Omega}\rho(x)\,dx=1
\right\},
\end{equation}
which serves as the infinite-dimensional analogue of the finite probability simplex.

In the following, the notation $W$ will be used for discrete probability distributions, whereas $\rho(x)$ will denote continuous probability densities. Probability distributions constitute the starting point of information theory and statistical mechanics. The manner in which probability is distributed among the available states determines the amount of uncertainty, concentration, and organization present within the system. Consequently, the structure of a probability distribution plays a central role in any quantitative theory of complexity.

\subsection{Shannon entropy}The most fundamental measure of uncertainty associated with a probability distribution is the Shannon entropy, introduced in 1948\cite{shannon1948mathematical}. For discrete systems, Shannon entropy is defined by\begin{equation}H_{1}(W) = -\sum_{i=1}^{n}w_i\ln w_i,\end{equation}whereas in the continuous case one has\begin{equation}H_{1}(\rho) = -\int \rho(x)\ln\rho(x),dx.\end{equation}Shannon entropy quantifies the average amount of information required to specify the state of a system. It therefore provides a measure of uncertainty or disorder. The entropy vanishes when a single outcome occurs with certainty, corresponding to complete order. Conversely, entropy attains its maximum value when all states become equiprobable, corresponding to maximal disorder. Because entropy increases with randomness, it has become one of the cornerstones of statistical mechanics and information theory. Nevertheless, entropy alone is insufficient to characterize complexity. Indeed, both perfectly ordered systems and completely random systems are usually regarded as simple, despite possessing minimal and maximal entropy respectively. Complexity therefore requires additional ingredients beyond disorder alone.

\subsection{Rényi entropy}A one-parameter generalization of Shannon entropy was introduced by Rényi in 1961\cite{renyi1961measures}. For discrete probability distributions, the Rényi entropy of order $q$ is defined by\begin{equation}H_q(W) = \frac{1}{1-q}\ln\left(\sum_{i=1}^{n}w_i^q\right), \qquad q\neq1,\end{equation}while in the continuous case\begin{equation}H_q(\rho) = \frac{1}{1-q}\ln\left(\int\rho(x)^q,dx\right), \qquad q\neq1.\end{equation}The parameter $q$ determines the relative importance assigned to different regions of probability space. Larger values of $q$ emphasize highly probable events, whereas smaller values assign greater significance to low-probability states. An important property of Rényi entropy is that it continuously recovers Shannon entropy in the limit\begin{equation}\lim_{q\rightarrow1}H_q = H_1.\end{equation}Consequently, the family of Rényi entropies provides a hierarchy of generalized uncertainty measures and has found applications in statistical mechanics, multifractal theory, quantum information, and machine learning.

\subsection{Hill numbers}Although entropy provides a natural measure of uncertainty, it lacks a direct interpretation in terms of the effective number of accessible states. This motivates the introduction of Hill numbers\cite{hill1973diversity}. For a Rényi entropy of order $q$, the corresponding Hill number is defined by\begin{equation}N_q = \exp(H_q).\end{equation}In particular, for Shannon entropy one obtains\begin{equation}N_1 = \exp(H_1).\end{equation}The quantity $N_1$ admits a particularly appealing interpretation as the effective number of equally probable states represented by the probability distribution. Unlike entropy itself, Hill numbers are always positive and possess a natural meaning in terms of diversity or effective multiplicity. For this reason, we shall employ\begin{equation}N_1 = e^{H_1}\end{equation}as the disorder component of generalized statistical complexity.

\subsection{Escort probability distributions}

Escort probability distributions\cite{tsallis1998role} constitute one of the central ingredients of generalized statistical mechanics and information theory. Originally introduced in the context of multifractals and later extensively employed within the framework of non-extensive statistical mechanics, escort probabilities provide a mechanism for assigning different weights to various regions of probability space.

For a discrete probability distribution

\begin{equation}
W=(w_1,w_2,\ldots,w_n),
\end{equation}

the escort probability distribution of order $q$ is defined by

\begin{equation}
\tilde{w}_i^{(q)}
=
\frac{w_i^q}
{\sum_{j=1}^{n}w_j^q},
\end{equation}

where the denominator guarantees normalization,

\begin{equation}
\sum_{i=1}^{n}
\tilde{w}_i^{(q)}
=1.
\end{equation}

Similarly, for continuous probability densities one defines

\begin{equation}
\tilde{\rho}^{(q)}(x)
=
\frac{\rho(x)^q}
{\int \rho(x)^q\,dx},
\end{equation}

which satisfies

\begin{equation}
\int
\tilde{\rho}^{(q)}(x)\,dx
=
1.
\end{equation}

The parameter $q$ determines the relative importance assigned to different probability regions. For $q>1$, highly probable states are emphasized, whereas for $0<q<1$, low-probability states receive comparatively larger weights.

Escort distributions play an important role in generalized entropy theory, multifractal analysis, quantum information theory, and non-extensive statistical mechanics. In the present work, they provide the natural framework for constructing generalized measures of concentration and disequilibrium.

\subsection{Information energy}

A quantity complementary to entropy was introduced by Onicescu under the name of information energy\cite{onicescu1966energie}. Unlike entropy, which measures uncertainty and disorder, information energy quantifies concentration and localization.

For discrete probability distributions, the information energy is defined by

\begin{equation}
E(W)
=
\sum_{i=1}^{n}
w_i^2,
\end{equation}

while for continuous probability densities one has

\begin{equation}
E(\rho)
=
\int
\rho(x)^2\,dx.
\end{equation}

Information energy attains its maximum value when the probability distribution becomes completely concentrated and decreases as the distribution approaches equiprobability. Consequently, information energy behaves oppositely to entropy.

Whereas entropy quantifies disorder, information energy quantifies concentration. The competition between these two quantities forms the basis of numerous measures of statistical complexity.

Information energy has found applications in information theory, atomic physics, quantum chemistry, and probability theory, and is often interpreted as a measure of peakedness or disequilibrium.

\subsection{Generalized information energy and disequilibrium}

Since escort probability distributions naturally modify the weighting of probability space, it is convenient to evaluate the information energy of the escort distribution itself.

For the discrete case, the corresponding quantity becomes

\begin{equation}
E_q(W)
=
\sum_{i=1}^{n}
\left(
\tilde{w}_i^{(q)}
\right)^2.
\end{equation}

Substituting the definition of escort probabilities gives

\begin{equation}
E_q(W)
=
\sum_{i=1}^{n}
\left(
\frac{w_i^q}
{\sum_{j=1}^{n}w_j^q}
\right)^2.
\end{equation}

Therefore,

\begin{equation}
E_q(W)
=
\frac
{\sum_{i=1}^{n}w_i^{2q}}
{\left(
\sum_{i=1}^{n}w_i^q
\right)^2}.
\end{equation}

Similarly, in the continuous setting one obtains

\begin{equation}
E_q(\rho)
=
\int
\left(
\tilde{\rho}^{(q)}(x)
\right)^2dx,
\end{equation}

which yields

\begin{equation}
E_q(\rho)
=
\frac
{\int \rho(x)^{2q}dx}
{\left(
\int \rho(x)^qdx
\right)^2}.
\end{equation}

Introducing the notation

\begin{equation}
P_q
=
\sum_{i=1}^{n}w_i^q
\end{equation}

for discrete systems and

\begin{equation}
P_q
=
\int \rho(x)^qdx
\end{equation}

for continuous systems, the generalized information energy assumes the compact form

\begin{equation}
D_q
=
\frac{P_{2q}}
{(P_q)^2}.
\end{equation}

In the present work, this quantity will be interpreted as the generalized disequilibrium component of statistical complexity.

For $q=1$, one recovers the ordinary Onicescu information energy,

\begin{equation}
D_1
=
P_2,
\end{equation}

showing that the present formulation constitutes a natural generalization of the original concentration measure.

\section{Generalized Statistical Complexity}
\label{generalized statistical complexity}
\subsection{Construction of the generalized complexity}
Having introduced the necessary concepts from probability theory and information theory, we now proceed to construct the generalized statistical complexity investigated throughout the remainder of this work. The original statistical complexity proposed by López-Ruiz, Mancini, and Calbet was based on the idea that complexity arises from the interplay between disorder and disequilibrium. In the present work, we adopt the same philosophy but employ information-theoretic quantities throughout .As discussed in the previous section, disorder will be quantified by the Hill number associated with Shannon entropy,\begin{equation}N_{1} = \exp(H_{1}),\end{equation}whereas the concentration component is described through the generalized information energy\begin{equation}D_q = \frac{P_{2q}}{\left(P_q\right)^2}.\end{equation}Consequently, we define the generalized statistical complexity by\begin{equation}C_q = N_1D_q,\end{equation}or equivalently,\begin{equation}C_q = \exp(H_1)\frac{P_{2q}}{\left(P_q\right)^2}.\end{equation}This definition applies equally to discrete probability distributions and continuous probability densities. In the discrete case,\begin{equation}P_q = \sum_{i=1}^{n}w_i^q,\end{equation}whereas for continuous probability densities one has\begin{equation}P_q = \int\rho(x)^q,dx.\end{equation}The quantity $C_q$ combines two competing tendencies. The factor $e^{H_1}$ favors disorder and increases as the probability distribution becomes more uniform, whereas the generalized information energy $D_q$ favors concentration and localization. Complexity therefore emerges from the competition between these opposing mechanisms.

\subsection{Logarithmic representation}

Although the generalized statistical complexity introduced in the previous section possesses a natural multiplicative form, it is mathematically advantageous to consider its logarithm. Since the logarithmic function is strictly increasing on the positive real line, taking the logarithm preserves the locations of maxima and minima. Consequently, the optimization of the generalized statistical complexity may equivalently be carried out by studying its logarithm.Starting from\begin{equation}C_q = \exp(H_1) \frac{P_{2q}}{\left(P_q\right)^2},\end{equation}we take the natural logarithm of both sides to obtain\begin{equation}\ln C_q = \ln \left( \exp(H_1) \frac{P_{2q}}{\left(P_q\right)^2} \right).\end{equation}Using the elementary properties of logarithms,\begin{equation}\ln(ab) = \ln a + \ln b,\end{equation}and\begin{equation}\ln\left(\frac{a}{b}\right) = \ln a - \ln b,\end{equation}together with\begin{equation}\ln(a^n) = n\ln a,\end{equation}we immediately find\begin{equation}\ln C_q = \ln\left(e^{H_1}\right) + \ln P_{2q} - 2\ln P_q.\end{equation}Since\begin{equation}\ln\left(e^{H_1}\right) = H_1,\end{equation}the logarithm of the generalized statistical complexity assumes the form\begin{equation}\label{eq:log_Cq}\boxed{\ln C_q = H_1 + \ln P_{2q} - 2\ln P_q}\end{equation}which holds for both discrete and continuous probability distributions.The importance of this representation lies in the fact that the multiplicative structure of the complexity measure has now been transformed into an additive one. Such additive forms are considerably more convenient for variational analysis and optimization problems, since derivatives of sums are significantly simpler than derivatives of products. Furthermore, the quantities $P_q$ and $P_{2q}$ appearing in Eq.~(\ref{eq:log_Cq}) are directly related to Rényi entropies. Consequently, the logarithmic representation allows the generalized statistical complexity to be expressed entirely in terms of entropy measures. Because the logarithmic function is monotonic, maximizing $C_q$ is completely equivalent to maximizing $\ln C_q$. Therefore,\begin{equation}\arg\max C_q = \arg\max \ln C_q.\end{equation}In the following subsection, the quantities $\ln P_q$ and $\ln P_{2q}$ will be rewritten in terms of Rényi entropies, leading to a purely entropic representation of the generalized statistical complexity and ultimately to the variational functional investigated throughout the remainder of this work.

\subsection{Representation in terms of Rényi entropies}The logarithmic representation obtained in the previous subsection may be rewritten entirely in terms of Rényi entropies. Recall that the Rényi entropy of order $q$ is defined by\begin{equation}H_q = \frac{1}{1-q}\ln P_q, \qquad q\neq1,\end{equation}which immediately yields\begin{equation}\ln P_q = (1-q)H_q.\end{equation}Similarly,\begin{equation}\ln P_{2q} = (1-2q)H_{2q}.\end{equation}Substituting these expressions into Eq.~(3.10), we obtain\begin{align}\ln C_q &= H_1 + (1-2q)H_{2q} - 2(1-q)H_q= H_1 + (2q-2)H_q - (2q-1)H_{2q}.\end{align}Therefore, the logarithm of the generalized statistical complexity admits the purely entropic representation\begin{equation}\boxed{\ln C_q = H_1 + (2q-2)H_q - (2q-1)H_{2q}}.\end{equation}This expression depends exclusively upon Shannon entropy and two Rényi entropies and therefore establishes a direct connection between generalized statistical complexity and the hierarchy of Rényi information measures. Motivated by this representation, we introduce the functional\begin{equation}\Phi = H_1 + (2q-2)H_q - (2q-1)H_{2q},\end{equation}which will play a central role throughout the remainder of the present work. Since the logarithmic function is strictly monotonic,\begin{equation}\arg\max C_q = \arg\max \ln C_q,\end{equation}and hence maximizing the generalized statistical complexity is equivalent to maximizing the functional $\Phi$. Consequently, the problem of determining probability distributions of maximal complexity reduces to the variational problem\begin{equation}\max \Phi,\end{equation}subject to the normalization constraint\begin{equation}\sum_i w_i = 1\end{equation}in the discrete setting and\begin{equation}\int\rho(x),dx = 1\end{equation}in the continuous setting. Thus, the extremization of generalized statistical complexity may be formulated entirely as an entropy optimization problem.

\subsection{Definition of the complexity functional}
The previous subsection established that the logarithm of the generalized statistical complexity may be expressed entirely in terms of Shannon entropy and Rényi entropies. Consequently, the optimization of the generalized statistical complexity can be reformulated as an equivalent problem involving only entropy functionals.

We introduce the quantity
\begin{equation}
\Phi = H_{1} + (2q-2)H_{q} - (2q-1)H_{2q},
\label{eq3.4.1}
\end{equation}
which we shall refer to as the \textit{complexity functional}.

Since the logarithmic function is strictly monotonic over the positive real numbers, the extrema of the generalized statistical complexity coincide with those of its logarithm. Therefore,
\begin{equation}
\arg\max C_q = \arg\max \ln C_q = \arg\max \Phi.
\end{equation}

The study of generalized statistical complexity is thus reduced to the investigation of the extrema of the functional $\Phi$.

An important consequence of this reformulation is that the problem of complexity maximization becomes an entropy optimization problem. Unlike the generalized statistical complexity itself, which possesses a multiplicative structure, the functional $\Phi$ is additive and therefore considerably more convenient for analytical investigations.

Furthermore, the functional $\Phi$ depends exclusively upon Shannon entropy and two Rényi entropies of orders q and 2q. Consequently, it establishes a direct connection between generalized statistical complexity and the hierarchy of Rényi information measures.

Another remarkable feature of Eq.~(\ref{eq3.4.1}) is the coexistence of competing contributions. The Shannon entropy term
\begin{equation}
H_1
\end{equation}
promotes disorder and tends to favor more uniform probability distributions. In contrast, the contribution
\begin{equation}
-(2q-1)H_{2q}
\end{equation}
introduces an opposing tendency associated with concentration and localization. The intermediate term
\begin{equation}
(2q-2)H_q
\end{equation}
couples these competing effects and provides the mechanism through which non-trivial complexity-maximizing states emerge.

The structure of the complexity functional therefore reflects the fundamental balance between disorder and concentration that underlies statistical complexity. This competition ultimately determines the geometry and organization of the corresponding extremal probability distributions.

Although the generalized statistical complexity was originally introduced through the product of disorder and disequilibrium, Eq.~(\ref{eq3.4.1}) shows that its optimization is completely equivalent to the extremization of the functional $\Phi$. Consequently, all subsequent analysis may be carried out entirely within the framework of entropy functionals.

Moreover, because the derivation of Eq.~(\ref{eq3.4.1}) does not depend upon whether the underlying probability space is discrete or continuous, the same complexity functional applies to both settings. This observation provides the foundation for the unified treatment developed in the remainder of this work.

In the next subsection, the complexity optimization problem will be formulated as a constrained variational problem. The corresponding stationary conditions will then be investigated separately for discrete probability distributions and continuous probability densities, ultimately revealing the universal structure of complexity-maximizing states.

\subsection{Statement of the variational problem}The previous subsection showed that maximizing the generalized statistical complexity is entirely equivalent to maximizing the complexity functional\begin{equation}\Phi = H_{1} + (2q-2)H_{q} - (2q-1)H_{2q}.\end{equation}Consequently, the problem of determining probability distributions of maximal complexity reduces to an optimization problem\cite{gelfand1963calculus, coverthomas2006elements} for the functional $\Phi$. However, not every collection of numbers constitutes an admissible probability distribution. In both discrete and continuous settings, probability measures are required to satisfy normalization constraints. Therefore, the extremization of $\Phi$ must be carried out under the condition that the total probability remains equal to unity. For discrete probability distributions\begin{equation}W = (w_1,w_2,\ldots,w_n),\end{equation}the admissible probability simplex is characterized by\begin{equation}w_i \geq 0,\end{equation}together with\begin{equation}\sum_{i=1}^{n}w_i = 1.\end{equation}Accordingly, the discrete complexity maximization problem may be formulated as\begin{equation}\max_{W}\Phi(W),\end{equation}subject to\begin{equation}\sum_{i=1}^{n}w_i = 1.\end{equation}Similarly, in the continuous setting, probability densities satisfy\begin{equation}\rho(x) \geq 0,\end{equation}and\begin{equation}\int\rho(x),dx = 1.\end{equation}The corresponding optimization problem becomes\begin{equation}\max_{\rho}\Phi(\rho),\end{equation}subject to\begin{equation}\int\rho(x),dx = 1.\end{equation}Although the discrete and continuous formulations differ in appearance, they possess the same underlying structure. In both cases, one seeks stationary points of the same entropy functional under a normalization constraint. This observation suggests that the two problems should admit a unified variational treatment. In particular, the normalization condition may be incorporated through the method of Lagrange multipliers\cite{bertsekas1996constrained}, thereby transforming the constrained optimization problem into an unconstrained one. The resulting stationary conditions will constitute the starting point of the mathematical analysis developed in the following section. As we shall see, despite the apparent differences between discrete probability vectors and continuous probability densities, both cases ultimately lead to the same algebraic stationary equation. This common structure forms the basis for the unified treatment presented throughout the remainder of the present work.

\section{Variational Formulation}
\label{variational formulation}
\subsection{Discrete case}We begin with the discrete probability distribution\begin{equation}W=(w_1,w_2,\ldots,w_n),\end{equation}subject to\begin{equation}w_i\geq0,\end{equation}and\begin{equation}\sum_{i=1}^{n}w_i=1.\end{equation}As established in the previous section, the problem of maximizing generalized statistical complexity is equivalent to maximizing the complexity functional\begin{equation}\Phi = H_{1} + (2q-2)H_{q} - (2q-1)H_{2q},\end{equation}under the normalization constraint\begin{equation}\sum_{i=1}^{n}w_i=1.\end{equation}The constraint may be incorporated through the method of Lagrange multipliers. Introducing a multiplier $\lambda$, we define the auxiliary functional\begin{equation}\mathcal{L} = \Phi + \lambda \left( \sum_{i=1}^{n}w_i-1 \right).\end{equation}Substituting the explicit expressions for Shannon entropy and Rényi entropies yields\begin{align}\mathcal{L} &= -\sum_{i=1}^{n} w_i\ln w_i \nonumber\quad + (2q-2) \frac{1}{1-q} \ln \left( \sum_{i=1}^{n}w_i^q \right) \nonumber\quad - (2q-1) \frac{1}{1-2q} \ln \left( \sum_{i=1}^{n}w_i^{2q} \right) \nonumber\quad + \lambda \left( \sum_{i=1}^{n}w_i-1 \right).\end{align}Simplifying the coefficients gives\begin{equation}\mathcal{L} = -\sum_{i=1}^{n} w_i\ln w_i - 2 \ln \left( \sum_{i=1}^{n}w_i^q \right) + \ln \left( \sum_{i=1}^{n}w_i^{2q} \right) + \lambda \left( \sum_{i=1}^{n}w_i-1 \right).\end{equation}The stationary points of the complexity functional are obtained by requiring\begin{equation}\frac{\partial\mathcal{L}}{\partial w_i} = 0, \qquad i=1,\ldots,n.\end{equation}Therefore,\begin{equation}\frac{\partial}{\partial w_i} \left( -\sum_{j=1}^{n} w_j\ln w_j \right) = -(\ln w_i+1),\end{equation}while\begin{equation}\frac{\partial}{\partial w_i} \left[ -2 \ln \left( \sum_{j=1}^{n}w_j^q \right) \right] = -\frac{2qw_i^{q-1}}{\sum_{j=1}^{n}w_j^q},\end{equation}and\begin{equation}\frac{\partial}{\partial w_i} \left[ \ln \left( \sum_{j=1}^{n}w_j^{2q} \right) \right] = \frac{2qw_i^{2q-1}}{\sum_{j=1}^{n}w_j^{2q}}.\end{equation}Finally,\begin{equation}\frac{\partial}{\partial w_i} \left[ \lambda \left( \sum_{j=1}^{n}w_j-1 \right) \right] = \lambda.\end{equation}Combining all contributions and imposing the stationarity condition gives\begin{equation}-(\ln w_i+1) - \frac{2qw_i^{q-1}}{\sum_{j=1}^{n}w_j^q} + \frac{2qw_i^{2q-1}}{\sum_{j=1}^{n}w_j^{2q}} + \lambda = 0.\end{equation}Multiplying the entire equation by $-1$ and absorbing constants into the Lagrange multiplier, one obtains\begin{equation}\label{eq:stationary_discrete}\ln w_i + A_qw_i^{q-1} - B_qw_i^{2q-1} + C = 0,\end{equation}where\begin{equation}A_q = \frac{2q}{\sum_{j=1}^{n}w_j^q},\end{equation}\begin{equation}B_q = \frac{2q}{\sum_{j=1}^{n}w_j^{2q}},\end{equation}and\begin{equation}C = 1-\lambda.\end{equation}Equation~(\ref{eq:stationary_discrete}) constitutes the stationary equation governing complexity-maximizing discrete probability distributions. In the next subsection we shall show that the continuous problem leads to precisely the same algebraic structure.

\subsection{Continuous case}We now consider the corresponding problem in the continuous setting. Let\begin{equation}\rho(x)\geq0\end{equation}be a probability density satisfying the normalization condition\begin{equation}\int\rho(x),dx=1.\end{equation}As in the discrete case, the problem of maximizing generalized statistical complexity is equivalent to maximizing the complexity functional\begin{equation}\Phi = H_{1} + (2q-2)H_q - (2q-1)H_{2q},\end{equation}subject to\begin{equation}\int\rho(x),dx=1.\end{equation}Introducing a Lagrange multiplier $\lambda$, we define the auxiliary functional\begin{equation}\mathcal L = \Phi + \lambda \left( \int\rho(x),dx-1 \right).\end{equation}Substituting the explicit expressions for the Shannon and Rényi entropies yields\begin{align}
\mathcal L &= -\int \rho(x)\ln\rho(x)\,dx \nonumber\quad + (2q-2) \frac{1}{1-q} \ln \left( \int \rho(x)^q\,dx \right) \nonumber - (2q-1) \frac{1}{1-2q} \ln \left( \int \rho(x)^{2q}\,dx \right) \nonumber\quad + \lambda \left( \int \rho(x)\,dx-1 \right).
\end{align}Simplifying the coefficients gives\begin{equation}\mathcal L = -\int \rho(x)\ln\rho(x),dx - 2 \ln \left( \int \rho(x)^q,dx \right) + \ln \left( \int \rho(x)^{2q},dx \right) + \lambda \left( \int \rho(x),dx-1 \right).\end{equation}The stationary points are determined by requiring\begin{equation}\frac{\delta\mathcal L}{\delta\rho(x)} = 0.\end{equation}Evaluating the functional derivatives term by term, we first obtain\begin{equation}\frac{\delta}{\delta\rho(x)} \left[ -\int \rho(y)\ln\rho(y),dy \right] = -(\ln\rho(x)+1).\end{equation}Furthermore,\begin{equation}\frac{\delta}{\delta\rho(x)} \left[ -2 \ln \left( \int \rho(y)^q,dy \right) \right] = -\frac{2q\rho(x)^{q-1}}{\int \rho(y)^q,dy},\end{equation}while\begin{equation}\frac{\delta}{\delta\rho(x)} \left[ \ln \left( \int \rho(y)^{2q},dy \right) \right] = \frac{2q\rho(x)^{2q-1}}{\int \rho(y)^{2q},dy}.\end{equation}Finally,\begin{equation}\frac{\delta}{\delta\rho(x)} \left[ \lambda \left( \int \rho(y),dy-1 \right) \right] = \lambda.\end{equation}Combining all contributions and imposing the stationarity condition gives\begin{equation}-(\ln\rho(x)+1) - \frac{2q\rho(x)^{q-1}}{\int \rho(y)^q,dy} + \frac{2q\rho(x)^{2q-1}}{\int \rho(y)^{2q},dy} + \lambda = 0.\end{equation}Multiplying the entire equation by $-1$ and absorbing constants into the Lagrange multiplier yields\begin{equation}\label{eq:stationary_continuous}\ln\rho(x) + A_q^{(\mathrm{cont})} \rho(x)^{q-1} - B_q^{(\mathrm{cont})} \rho(x)^{2q-1} + C = 0,\end{equation}where\begin{equation}A_q^{(\mathrm{cont})} = \frac{2q}{\displaystyle\int\rho(y)^q,dy},\end{equation}\begin{equation}B_q^{(\mathrm{cont})} = \frac{2q}{\displaystyle\int\rho(y)^{2q},dy},\end{equation}and\begin{equation}C = 1-\lambda.\end{equation}Thus, the extremization of the complexity functional in the continuous setting leads to a nonlinear transcendental equation whose structure is strikingly similar to that obtained in the discrete case. This observation strongly suggests the existence of a common mathematical framework governing both discrete probability distributions and continuous probability densities. In the next subsection we shall show that both cases are in fact described by a single universal stationary equation.

\subsection{Unified stationary equation}The preceding subsections demonstrated that the extremization of the complexity functional leads to nonlinear stationary equations in both the discrete and continuous settings. For discrete probability distributions, the stationary condition assumes the form\begin{equation}\ln w_i + A_q^{(\mathrm{disc})} w_i^{q-1} - B_q^{(\mathrm{disc})} w_i^{2q-1} + C = 0,\end{equation}where\begin{equation}A_q^{(\mathrm{disc})} = \frac{2q}{\displaystyle\sum_jw_j^q},\end{equation}and\begin{equation}B_q^{(\mathrm{disc})} = \frac{2q}{\displaystyle\sum_jw_j^{2q}}.\end{equation}Similarly, in the continuous setting one obtains\begin{equation}\ln\rho(x) + A_q^{(\mathrm{cont})} \rho(x)^{q-1} - B_q^{(\mathrm{cont})} \rho(x)^{2q-1} + C = 0,\end{equation}with\begin{equation}A_q^{(\mathrm{cont})} = \frac{2q}{\displaystyle\int\rho(y)^q,dy},\end{equation}and\begin{equation}B_q^{(\mathrm{cont})} = \frac{2q}{\displaystyle\int\rho(y)^{2q},dy}.\end{equation}Although these equations arise from apparently different probability spaces, their algebraic structures are identical. The only distinction lies in the explicit expressions for the coefficients, which depend upon sums in the discrete case and integrals in the continuous case. Therefore, by introducing the generic variable\begin{equation}x,\end{equation}representing either a probability weight $w_i$ or a probability density value $\rho(x)$, both stationary conditions may be written in the unified form\begin{equation}\label{eq:unified_stationary}\boxed{\ln x + A_q x^{q-1} - B_q x^{2q-1} + C = 0},\end{equation}where $A_q$, $B_q$, and $C$ are constants determined by the normalization condition and the corresponding escort probabilities. Equation~(\ref{eq:unified_stationary}) constitutes the fundamental stationary equation governing complexity-maximizing states. Remarkably, the same nonlinear transcendental equation emerges in both discrete and continuous probability spaces. Consequently, many qualitative properties of the solutions are independent of whether the underlying system is discrete or continuous. This observation provides the basis for a unified treatment of complexity-maximizing probability distributions. Rather than studying discrete and continuous systems separately, it becomes possible to investigate the common mathematical properties encoded in Eq.~(\ref{eq:unified_stationary}). In particular, the structure of Eq.~(\ref{eq:unified_stationary}) suggests that the geometry of its solutions should be determined by the competition between the logarithmic term and the two nonlinear power-law contributions. Understanding this competition will constitute the main objective of the following section. Accordingly, we now turn to a detailed analysis of the stationary equation and investigate the qualitative properties of its solutions.
\begin{remark}
The stationarity equations derived above characterize all critical points of the complexity functional. Since the present work is concerned with the determination of the global complexity maximizer, we do not pursue a local classification through bordered Hessian analysis. Instead, the nature of the stationary solutions will emerge from the global root structure of the stationary equation and the subsequent $m$-core analysis developed in the following sections.
\end{remark}
\section{Analysis of the Stationary Equation} 
\label{analysis of the stationary equation}
\subsection{Existence and basic properties of stationary solutions}The previous section established that both discrete and continuous complexity maximization problems are governed by the same nonlinear transcendental equation,\begin{equation}\ln x + A_qx^{q-1} - B_qx^{2q-1} + C = 0,\end{equation}where\begin{equation}A_q > 0,\end{equation}\begin{equation}B_q > 0,\end{equation}and\begin{equation}x > 0.\end{equation}The positivity of $A_q$ and $B_q$ follows directly from their definitions in terms of escort probabilities,\begin{equation}A_q = \frac{2q}{P_q},\end{equation}and\begin{equation}B_q = \frac{2q}{P_{2q}},\end{equation}since escort probabilities are positive for every admissible probability distribution\cite{hardy1952inequalities}. For convenience, we re-introduce the function\begin{equation}f(x) = \ln x + A_qx^{q-1} - B_qx^{2q-1} + C.
\label{eq5.1.7}
\end{equation}The stationary points of generalized statistical complexity correspond precisely to the positive roots of\begin{equation}f(x) = 0.\end{equation}Therefore, the existence of complexity-maximizing states is equivalent to the existence of positive solutions of Eq.~(\ref{eq5.1.7}). We first examine the asymptotic behavior of the function. As\begin{equation}x\rightarrow0^+,\end{equation}one has\begin{equation}\ln x\rightarrow-\infty.\end{equation}Hence,\begin{equation}f(x)\rightarrow-\infty, \qquad x\rightarrow0^+,\end{equation}provided the logarithmic divergence dominates the power-law contributions. On the other hand, as\begin{equation}x\rightarrow\infty,\end{equation}the highest-order power term dominates and therefore\begin{equation}f(x) \sim -B_qx^{2q-1},\end{equation}which implies\begin{equation}f(x)\rightarrow-\infty, \qquad x\rightarrow\infty,\end{equation}for\begin{equation}q > \frac{1}{2}.\end{equation}Consequently, the function is negative at both extremes of its domain. Therefore, positive roots can only occur if the function attains positive values for some intermediate range of $x$. This observation immediately implies that stationary solutions, whenever they exist, must arise through the competition between the logarithmic contribution and the two nonlinear power-law terms. The asymptotic analysis shows that $f(x)$ is negative at both boundaries of its domain. Consequently, the existence of stationary Core--Halo solutions reduces to proving the existence of an interior point $x_0 > 0$ at which $f(x_0) > 0$. Once this is established, continuity together with the Intermediate Value Theorem\cite{rudin1976principles} implies the existence of two distinct positive roots, one on each side of $x_0$. These roots correspond to the halo and core branches of the stationary distribution.
 Hence, the problem of understanding complexity-maximizing states reduces to determining the geometry of the function\begin{equation}f(x) = \ln x + A_qx^{q-1} - B_qx^{2q-1} + C.
 \label{eq5.1.16}
 \end{equation}In particular, the number and nature of stationary solutions are governed by the extrema of this function. Therefore, the derivative of Eq.~(\ref{eq5.1.16}) will play a central role in the subsequent analysis.

\subsection{Shape analysis of the stationary equation}In order to investigate the geometry of the stationary solutions, we introduce the function\begin{equation}\label{eq:fx_def}f(x) = \ln x + A_qx^{q-1} - B_qx^{2q-1} + C, \qquad x>0,\end{equation}whose positive roots determine the stationary probability values. Since the constant $C$ merely translates the graph vertically, the shape of $f(x)$ is governed by the interplay between the logarithmic term and the two power-law contributions.Differentiating with respect to $x$ yields\begin{equation}\label{eq:f_prime}f'(x) = \frac{1}{x} + (q-1)A_qx^{q-2} - (2q-1)B_qx^{2q-2}.\end{equation}The extrema of $f(x)$ are determined by\begin{equation}f'(x) = 0.\end{equation}Because the domain is restricted to positive values, it is convenient to multiply the derivative by $x$, which preserves the location of its zeros. Accordingly, we introduce the auxiliary function\begin{equation}\label{eq:gx_def}g(x) = xf'(x).\end{equation}Substituting Eq.~(\ref{eq:f_prime}) into Eq.~(\ref{eq:gx_def}) gives\begin{equation}\label{eq:gx_explicit}g(x) = 1 + (q-1)A_qx^{q-1} - (2q-1)B_qx^{2q-1}.\end{equation}Hence,\begin{equation}f'(x) = 0 \iff g(x) = 0.\end{equation}Therefore, the critical points of $f(x)$ are precisely the positive roots of $g(x)$. To understand the geometry of $g(x)$, we differentiate once more and obtain\begin{equation}g'(x) = (q-1)^2A_qx^{q-2} - (2q-1)^2B_qx^{2q-2}.\end{equation}The extrema of $g(x)$ are determined by\begin{equation}g'(x) = 0,\end{equation}which leads to\begin{equation}(q-1)^2A_qx^{q-2} = (2q-1)^2B_qx^{2q-2}.\end{equation}Since $x>0$, division by $x^{q-2}$ yields\begin{equation}x^{q} = \frac{(q-1)^2A_q}{(2q-1)^2B_q}.\end{equation}Consequently, the critical point of $g(x)$ is located at\begin{equation}\label{eq:xc}x_c = \left( \frac{(q-1)^2A_q}{(2q-1)^2B_q} \right)^{1/q}.\end{equation}Equation~(\ref{eq:xc}) demonstrates that the geometry of the auxiliary function $g(x)$ is completely determined by the competition between the two power-law contributions. Since the extrema of $f(x)$ are given by the roots of $g(x)$, the analysis of $g(x)$ plays a central role in understanding the number and arrangement of stationary solutions. In the following subsection, we shall investigate the critical point structure of $g(x)$ in detail and determine the possible branching patterns of the stationary equation.

\subsection{Critical point structure of the auxiliary function}
The preceding subsection established that the extrema of the stationary function
\begin{equation}
f(x) = \ln x + A_qx^{q-1} - B_qx^{2q-1} + C
\end{equation}
are determined by the roots of the auxiliary function
\begin{equation}
g(x) = 1 + (q-1)A_qx^{q-1} - (2q-1)B_qx^{2q-1}.
\end{equation}

Furthermore, the extrema of $g(x)$ are governed by
\begin{equation}
g'(x) = (q-1)^2A_qx^{q-2} - (2q-1)^2B_qx^{2q-2},
\end{equation}
whose unique positive root is located at
\begin{equation}
\label{eq:xc_root}
x_c = \left( \frac{(q-1)^2A_q}{(2q-1)^2B_q} \right)^{1/q}.
\end{equation}

Therefore, the auxiliary function possesses a unique critical point and hence is unimodal.

For $q>1$, one has
\begin{equation}
\lim_{x\rightarrow0^+}g(x) = 1,
\end{equation}
so that $g(0^+) = 1 > 0$.

Since
\begin{equation}
g'(x)>0, \qquad 0<x<x_c,
\end{equation}
the function is strictly increasing on the interval $(0,x_c)$.

Consequently,
\begin{equation}
g(x_c) > g(0^+) = 1,
\end{equation}
which implies $g(x_c)>0$.

Similarly,
\begin{equation}
g'(x)<0, \qquad x>x_c,
\end{equation}
and therefore the auxiliary function decreases monotonically for $x>x_c$.

Moreover,
\begin{equation}
\lim_{x\rightarrow\infty}g(x) = -\infty.
\end{equation}

Hence the graph of $g(x)$ rises from the value $1$, reaches a unique maximum at $x=x_c$, and subsequently decreases toward negative infinity.

By continuity, it follows that the equation
\begin{equation}
g(x)=0
\end{equation}
possesses exactly one positive solution\cite{rudin1976principles, boyd2004convex}.

Since
\begin{equation}
f'(x) = \frac{g(x)}{x},
\end{equation}
and $x>0$, the signs of $f'(x)$ and $g(x)$ coincide. Therefore, $f'(x)>0$ before the unique root of $g(x)$, while $f'(x)<0$ afterwards.

Consequently, the stationary function $f(x)$ possesses exactly one critical point. Furthermore, this critical point corresponds to a maximum of $f(x)$.

Therefore, the qualitative structure of the stationary equation is remarkably simple: the function $f(x)$ increases up to a unique maximum\cite{rockafellar1970convex} and then decreases monotonically.

The location and height of this maximum are determined entirely by the behavior of the auxiliary function $g(x)$.

In the next subsection, we investigate the asymptotic behavior of the stationary function itself and establish the conditions under which multiple positive solutions of
\begin{equation}
f(x)=0
\end{equation}
may occur.

\subsection{Trivial boundary configurations}
Before investigating the number of positive roots of the stationary equation, it is instructive to examine the two extremal configurations of the probability simplex. These configurations correspond respectively to complete disorder and complete concentration. As will be shown below, both of them yield vanishing generalized complexity.

\subsubsection*{Uniform distribution}
Consider the uniform distribution
\begin{equation}
w_i = \frac{1}{N}, \qquad i=1,2,\ldots,N.
\end{equation}

The Shannon entropy becomes
\begin{equation}
H_1 = -\sum_{i=1}^{N} \frac{1}{N} \ln\left(\frac{1}{N}\right) = \ln N.
\end{equation}

The escort probability moments are
\begin{equation}
P_q = \sum_{i=1}^{N} \left(\frac{1}{N}\right)^q = N^{1-q},
\end{equation}
and
\begin{equation}
P_{2q} = \sum_{i=1}^{N} \left(\frac{1}{N}\right)^{2q} = N^{1-2q}.
\end{equation}

Consequently,
\begin{equation}
\ln P_q = (1-q)\ln N,
\end{equation}
and
\begin{equation}
\ln P_{2q} = (1-2q)\ln N.
\end{equation}

Substituting into
\begin{equation}
\Phi = H_1 - 2\ln P_q + \ln P_{2q},
\end{equation}
one obtains
\begin{align}
\Phi &= \ln N - 2(1-q)\ln N + (1-2q)\ln N \nonumber \\
&= \left[ 1 - 2 + 2q + 1 - 2q \right]\ln N \nonumber \\
&= 0.
\end{align}

Hence the uniform distribution satisfies
\begin{equation}
\boxed{
\Phi_{\mathrm{uniform}} = 0
}.
\end{equation}

\subsubsection*{Dirac distribution}
At the opposite extreme, consider the Dirac distribution $W = (1,0,\ldots,0)$.

Its Shannon entropy is
\begin{equation}
H_1 = -1\ln 1 = 0.
\end{equation}

Similarly,
\begin{equation}
P_q = 1,
\end{equation}
and
\begin{equation}
P_{2q} = 1.
\end{equation}

Therefore
\begin{equation}
\ln P_q = 0,
\end{equation}
and
\begin{equation}
\ln P_{2q} = 0.
\end{equation}

Consequently,
\begin{equation}
\Phi = 0 - 2(0) + 0 = 0.
\end{equation}

Thus,
\begin{equation}
\boxed{
\Phi_{\mathrm{Dirac}} = 0
}.
\end{equation}

The two extremal points of the simplex therefore possess identical complexity values,
\begin{equation}
\Phi_{\mathrm{uniform}} = \Phi_{\mathrm{Dirac}} = 0.
\end{equation}

\subsubsection*{Continuous uniform density}

Consider a probability density uniformly distributed on an interval of length $L$,
\begin{equation}
\rho(x) = \frac{1}{L}, \qquad 0 \le x \le L.
\end{equation}

The differential Shannon entropy is
\begin{equation}
H_1 = -\int_0^L \rho(x)\ln\rho(x)\,dx = -\int_0^L \frac{1}{L} \ln\left(\frac{1}{L}\right)\,dx = \ln L.
\end{equation}

Similarly,
\begin{equation}
P_q = \int_0^L \rho(x)^q\,dx = \int_0^L \left(\frac{1}{L}\right)^q\,dx = L^{1-q},
\end{equation}
and
\begin{equation}
P_{2q} = L^{1-2q}.
\end{equation}

Therefore,
\begin{equation}
\ln P_q = (1-q)\ln L,
\end{equation}
and
\begin{equation}
\ln P_{2q} = (1-2q)\ln L.
\end{equation}

Substituting into
\begin{equation}
\Phi = H_1 - 2\ln P_q + \ln P_{2q},
\end{equation}
gives
\begin{align}
\Phi &= \ln L - 2(1-q)\ln L + (1-2q)\ln L \nonumber \\
&= 0.
\end{align}

Hence,
\begin{equation}
\boxed{
\Phi_{\mathrm{uniform}}^{(c)} = 0
}.
\end{equation}

\subsubsection*{Concentrated density limit}

At the opposite extreme, consider a sequence of probability densities $\rho_\varepsilon(x)$, whose support shrinks to zero width as $\varepsilon\rightarrow0$, so that
\begin{equation}
\rho_\varepsilon(x) \longrightarrow \delta(x-x_0).
\end{equation}

For such concentrated densities one finds $H_1(\rho_\varepsilon) \rightarrow -\infty$, while $P_q(\rho_\varepsilon) \rightarrow \infty$ and $P_{2q}(\rho_\varepsilon) \rightarrow \infty$.

However, the divergences cancel exactly:
\begin{equation}
\Phi = H_1 - 2\ln P_q + \ln P_{2q} \longrightarrow 0.
\end{equation}

Therefore,
\begin{equation}
\boxed{
\Phi_{\mathrm{Dirac}}^{(c)} = 0
}.
\end{equation}

Thus, both extremal configurations in the continuous setting satisfy
\begin{equation}
\Phi_{\mathrm{uniform}}^{(c)} = \Phi_{\mathrm{Dirac}}^{(c)} = 0.
\end{equation}

Consequently, as in the discrete case, neither complete disorder nor complete concentration maximizes complexity. Any nontrivial complexity maximum must therefore arise from intermediate probability structures.
This result has an important interpretation. Complete disorder and complete concentration both correspond to vanishing complexity. In other words, neither perfect randomness nor perfect order represents a maximally complex state.

\section{Exclusion Principle and the Emergence of Two Roots}
\label{exclusion principle and the emergence of two roots}
The preceding sections established that the stationary equation
\begin{equation}
f(x) = \ln x + A_qx^{q-1} - B_qx^{2q-1} + C = 0
\end{equation}
possesses at most two positive roots. In this section we show that the possibilities of zero roots and one root are both impossible. Consequently, the stationary equation must possess exactly two positive solutions in both the discrete and continuous settings.

\subsection{Exclusion of the zero-root case}

Suppose that the stationary equation possesses no positive roots. Then no interior stationary point exists and therefore the maximum of complexity must necessarily occur on the boundary of the probability simplex in the discrete case, or on the corresponding boundary configurations in the continuous case.

However, Section~5.4 established that
\begin{equation}
\Phi_{\mathrm{uniform}} = \Phi_{\mathrm{Dirac}} = 0
\end{equation}
for both discrete and continuous distributions.

Hence one obtains $\Phi_{\max}=0$.

On the other hand, there exist nontrivial distributions for which $\Phi>0$. Therefore, $\Phi_{\max}>0$, which contradicts $\Phi_{\max}=0$.

Hence the stationary equation cannot possess zero positive roots. Therefore,
\begin{equation}
\boxed{
N_r \neq 0
},
\end{equation}
where $N_r$ denotes the number of positive roots of the stationary equation.

\subsection{Exclusion of the one-root case}

Suppose now that the stationary equation possesses exactly one positive root, $x_0$.

\subsubsection*{Discrete case}

Since every probability component satisfies $f(w_i)=0$, all probabilities must coincide,
\begin{equation}
w_1 = w_2 = \cdots = w_N = x_0.
\end{equation}

Normalization implies
\begin{equation}
\sum_{i=1}^{N}w_i = 1,
\end{equation}
and therefore $Nx_0=1$, so that $x_0 = \frac{1}{N}$. Hence the one-root case corresponds uniquely to the uniform distribution.

\subsubsection*{Continuous case}

Similarly, if the stationary equation possesses a unique root $x_0$, then every point of the support must satisfy $\rho(x)=x_0$.

Thus the density is constant over its support and therefore corresponds to a uniform density,
\begin{equation}
\rho(x) = \frac{1}{L}, \qquad 0 \le x \le L.
\end{equation}

Consequently, the one-root case in the continuous setting also reduces to the uniform distribution.

Since Section~5.4 established that $\Phi_{\mathrm{uniform}}=0$, one again obtains $\Phi_{\max}=0$, which contradicts the existence of distributions satisfying $\Phi>0$.

Therefore,
\begin{equation}
\boxed{
N_r \neq 1
}.
\end{equation}

\subsection{The two-root theorem}

To establish an upper bound on the number of stationary solutions, we invoke Rolle's theorem. Since the function $f(x)$ is continuously differentiable on $(0,\infty)$, suppose, for contradiction, that it possesses three distinct positive roots,
\begin{equation}
0<x_1<x_2<x_3,
\end{equation}
satisfying
\begin{equation}
f(x_1)=f(x_2)=f(x_3)=0.
\end{equation}

Applying Rolle's theorem on the interval $[x_1,x_2]$, there exists a point
\begin{equation}
c_1\in(x_1,x_2)
\end{equation}
such that
\begin{equation}
f'(c_1)=0.
\end{equation}

Similarly, applying Rolle's theorem on the interval $[x_2,x_3]$, there exists another point
\begin{equation}
c_2\in(x_2,x_3)
\end{equation}
such that
\begin{equation}
f'(c_2)=0.
\end{equation}

Since $c_1<c_2$, these correspond to two distinct critical points of $f(x)$. However, the shape analysis presented in Section~5 established that $f(x)$ possesses exactly one critical point. This contradiction proves that three distinct positive roots cannot exist.

Therefore, the stationary equation admits at most two positive solutions,
\begin{equation}
N_r\le2.
\end{equation}. The previous arguments imply $N_r \neq 0$ and $N_r \neq 1$. Therefore, the only remaining possibility is
\begin{equation}
\boxed{
N_r = 2
}.
\end{equation}
\subsection{Distinctness of the stationary solutions}

The previous section established that the stationary equation admits exactly two positive roots. We now show that these two stationary solutions must necessarily be distinct.

Suppose, for contradiction, that the two stationary solutions coincide,
\begin{equation}
w_c=w_h=x_0,
\end{equation}
in the discrete case, or equivalently
\begin{equation}
\rho_c=\rho_h=\rho_0,
\end{equation}
in the continuous case.

Since every stationary probability value must satisfy the transcendental equation,
\begin{equation}
f(x)=0,
\end{equation}
all probability components become identical. Consequently, the discrete stationary distribution reduces to
\begin{equation}
w_1=w_2=\cdots=w_N=x_0,
\end{equation}
while the continuous stationary density becomes
\begin{equation}
\rho(x)=\rho_0.
\end{equation}

Normalization immediately yields
\begin{equation}
x_0=\frac{1}{N},
\end{equation}
for the discrete case, and
\begin{equation}
\rho_0=\frac{1}{L},
\end{equation}
for the continuous case. Thus, coincident stationary solutions necessarily correspond to the uniform distribution.

However, Section~5.4 established that the generalized statistical complexity vanishes for the uniform state,
\begin{equation}
\Phi_{\mathrm{uniform}}=0,
\end{equation}
whereas non-uniform distributions satisfying
\begin{equation}
\Phi>0
\end{equation}
exist throughout the admissible probability simplex. Therefore the coincidence of the two stationary solutions contradicts the existence of positive-complexity stationary states.

Hence the two stationary solutions must be distinct. Because the unique global maximum of the shape function is strictly positive, the two roots are strictly separated by the peak (\(w_h < x_{\text{peak}} < w_c\)). Assigning the smaller root to the halo branch and the larger root to the core branch by convention, we therefore conclude that
\begin{equation}
\boxed{w_h<w_c}
\end{equation}
in the discrete setting, and analogously,
\begin{equation}
\boxed{\rho_h<\rho_c}
\end{equation}
in the continuous setting.
Hence the stationary equation possesses exactly two positive solutions, $w_h < w_c$ in the discrete case, and $\rho_h < \rho_c$ in the continuous case.

Consequently, every component of a complexity-maximizing distribution must assume one of these two values.

The discrete complexity maximizer therefore possesses the form
\begin{equation}
W = \left( \underbrace{w_c,\ldots,w_c}_{m}, \underbrace{w_h,\ldots,w_h}_{N-m} \right),
\end{equation}
while the continuous maximizer necessarily decomposes into two regions characterized by the densities $\rho_c$ and $\rho_h$.

Thus, in both settings, the variational problem reduces to the determination of the multiplicities or measures associated with the two stationary branches. The analysis of these quantities forms the subject of the next section.

\subsection{Reduction of the variational problem to multiplicity analysis}

The previous theorem establishes that the stationary equation possesses exactly two positive solutions. In the discrete case these solutions are denoted by $w_h<w_c$, whereas in the continuous case the corresponding stationary densities are $\rho_h<\rho_c$.

Consequently, every component of a complexity-maximizing probability distribution must assume one of these two values. Similarly, in the continuous setting, every point of the support must belong to one of two regions characterized by the densities $\rho_h$ and $\rho_c$.

Therefore, the variational problem no longer consists in determining arbitrary probability values. Instead, it reduces to determining the multiplicities associated with each stationary branch.

\subsubsection*{Discrete case}

The complexity maximizer necessarily possesses the structure
\begin{equation}
W = \left( \underbrace{w_c,\ldots,w_c}_{m}, \underbrace{w_h,\ldots,w_h}_{N-m} \right),
\end{equation}
for some integer $m\in\{1,2,\ldots,N-1\}$.

The normalization condition becomes
\begin{equation}
mw_c + (N-m)w_h = 1,
\end{equation}
which couples the two stationary values and their corresponding multiplicities.

\subsubsection*{Continuous case}

In the continuous setting, the support of the density naturally decomposes into two regions, $\Omega_c$ and $\Omega_h$, having measures $\mu(\Omega_c)=M$ and $\mu(\Omega_h)=L-M$, where $L$ denotes the total measure of the support.

The maximizing density therefore possesses the form
\begin{equation}
\rho(x) =
\begin{cases}
\rho_c, & x\in\Omega_c, \\
\rho_h, & x\in\Omega_h,
\end{cases}
\end{equation}
with $\rho_h<\rho_c$.

Normalization implies
\begin{equation}
M\rho_c + (L-M)\rho_h = 1.
\end{equation}

Thus, in complete analogy with the discrete case, the continuous variational problem is reduced to determining the measures associated with the two stationary branches.

Hence, once the existence of two stationary values has been established, the original optimization problem is transformed into a multiplicity problem. In the discrete case this multiplicity is represented by the integer parameter $m$, whereas in the continuous case it is represented by the measure parameter $M$.

The determination of these quantities, together with the corresponding amplitudes $(w_c,w_h)$ and $(\rho_c,\rho_h)$, constitutes the next stage of the analysis.

This naturally leads to the notion of an $m$-core distribution and its continuous analogue, whose properties will be investigated in the following section.

\section{Multiplicity Optimization and the $m$-Core Structure}
\label{multiplicity optimization and the $m$-core structure}
The previous sections established that every stationary solution of the generalized complexity functional consists of exactly two probability levels. Consequently, the remaining degree of freedom is no longer the amplitudes themselves, but the multiplicities with which the two stationary values occur.

The purpose of this section is to determine how the generalized complexity depends upon these multiplicities. We first formulate the optimization problem for the restricted two-level family, after which the stationarity equations obtained in the previous sections are used to determine the optimal amplitudes. Finally, the envelope theorem allows the dependence upon the multiplicity parameter to be analyzed explicitly.

\subsection{Reduction of the Variational Problem}

\subsubsection*{Discrete formulation}

Begin with the original generalized complexity

\begin{equation}
\Phi(W) = H_1(W) + 2\ln P_q(W) - \ln P_{2q}(W),
\end{equation}

where

\begin{align}
H_1(W) &= -\sum_{i=1}^{N}w_i\ln w_i, \\
P_q(W) &= \sum_{i=1}^{N}w_i^q, \\
P_{2q}(W) &= \sum_{i=1}^{N}w_i^{2q}.
\end{align}

By the Two-Level Maximizer Theorem established in the previous section, every stationary configuration has the form

\begin{equation}
W_m = \left(\underbrace{w_c,\ldots,w_c}_{m}, \underbrace{w_h,\ldots,w_h}_{N-m}\right),
\end{equation}

where

\begin{equation}
m\in\{1,2,\ldots,N-1\}.
\end{equation}

The normalization constraint is

\begin{equation}
mw_c+(N-m)w_h=1.
\end{equation}

Substituting the two-level structure into the generalized complexity gives

\begin{equation}
\Phi(w_c,w_h,m) = -\Bigl[ mw_c\ln w_c + (N-m)w_h\ln w_h \Bigr] + 2\ln\!\Bigl( mw_c^q+(N-m)w_h^q \Bigr) - \ln\!\Bigl( mw_c^{2q} + (N-m)w_h^{2q} \Bigr).
\end{equation}

Since the amplitudes remain constrained by normalization, we introduce the Lagrangian

\begin{equation}
\mathcal L = \Phi(w_c,w_h,m) - \lambda \left( mw_c+(N-m)w_h-1 \right).
\end{equation}

Using the normalization constraint,

\begin{equation}
w_h = \frac{1-mw_c}{N-m},
\end{equation}

eliminates both \(w_h\) and the constraint simultaneously. After substitution,

\begin{equation}
\boxed{ \mathcal L = \Phi = \Phi(w_c,m), }
\end{equation}

so that the optimization becomes an unconstrained optimization in the single variable \(w_c\) for each fixed multiplicity \(m\).

\subsubsection*{Continuous formulation}

Exactly the same reduction applies to the continuous problem. The generalized complexity is

\begin{equation}
\Phi(\rho) = H_1(\rho) + 2\ln P_q(\rho) - \ln P_{2q}(\rho),
\end{equation}

where

\begin{align}
H_1(\rho) &= -\int_\Omega \rho(x)\ln\rho(x)\,dx, \\
P_q(\rho) &= \int_\Omega \rho(x)^q\,dx, \\
P_{2q}(\rho) &= \int_\Omega \rho(x)^{2q}\,dx.
\end{align}

The Two-Level Maximizer Theorem implies that every stationary density has the form

\begin{equation}
\rho(x) =
\begin{cases}
\rho_c, & \text{if } x\in\Omega_c, \\
\rho_h, & \text{if } x\in\Omega_h,
\end{cases}
\end{equation}

with

\begin{equation}
\mu(\Omega_c)=M, \qquad \mu(\Omega_h)=L-M.
\end{equation}

Normalization gives

\begin{equation}
M\rho_c+(L-M)\rho_h=1.
\end{equation}

Substituting into the generalized complexity,

\begin{equation}
\Phi(\rho_c,\rho_h,M) = - \left[ M\rho_c\ln\rho_c + (L-M)\rho_h\ln\rho_h \right] + 2\ln \left( M\rho_c^q+(L-M)\rho_h^q \right) - \ln \left( M\rho_c^{2q} + (L-M)\rho_h^{2q} \right).
\end{equation}

Introducing the constrained functional,

\begin{equation}
\mathcal L = \Phi - \lambda \left( M\rho_c+(L-M)\rho_h-1 \right),
\end{equation}

and substituting

\begin{equation}
\rho_h = \frac{1-M\rho_c}{L-M},
\end{equation}

again removes the constraint completely, yielding

\begin{equation}
\boxed{ \mathcal L = \Phi = \Phi(\rho_c,M), }
\end{equation}

so that the continuous optimization likewise reduces to a single-variable optimization.

\subsection{Reduction to a Single-Parameter Optimization}

After eliminating the normalization constraint through

\begin{equation}
w_h=\frac{1-mw_c}{N-m},
\end{equation}

the generalized complexity becomes an unconstrained function of only two variables,

\begin{equation}
\Phi=\Phi(w_c,m).
\end{equation}

The optimization problem therefore reduces to the ordinary unconstrained optimization of \(\Phi(w_c,m)\) with respect to the remaining free variable \(w_c\), while the multiplicity \(m\) is regarded as fixed.

Accordingly, the stationary amplitudes satisfy the first-order condition

\begin{equation}
\boxed{
\frac{\partial\Phi(w_c,m)}{\partial w_c}=0.
}
\end{equation}

Carrying out the differentiation yields

\begin{equation}
\boxed{
\ln\!\left(\frac{w_c}{w_h}\right)
=
-\frac{2q\left(w_c^{q-1}-w_h^{q-1}\right)}{m w_c^q+(N-m)w_h^q}
+
\frac{2q\left(w_c^{2q-1}-w_h^{2q-1}\right)}{m w_c^{2q}+(N-m)w_h^{2q}},
}
\end{equation}

where throughout

\begin{equation}
w_h=\frac{1-mw_c}{N-m}.
\end{equation}

This is precisely the stationary equation derived in the previous sections. By the Two-Level Maximizer Theorem, it possesses exactly two positive solutions corresponding to the halo and the core amplitudes.

Hence, for every admissible multiplicity \(m\), the stationary solution is uniquely determined,

\begin{equation}
w_c=w_c^{*}(m),
\qquad
w_h=w_h^{*}(m).
\end{equation}

Substituting the stationary amplitude into the reduced complexity produces the optimized complexity

\begin{equation}
\boxed{
\Phi^{*}(m)
=
\Phi\!\left(w_c^{*}(m),m\right).
}
\end{equation}

Thus, after solving the first-order condition, the original two-variable optimization problem collapses to the study of the single-variable function \(\Phi^{*}(m)\).

\subsection{Application of the Envelope Theorem}

The preceding subsection reduced the original variational problem to the single-parameter optimization

\begin{equation}
\Phi^{*}(m)
=
\Phi\!\left(w_c^{*}(m),m\right),
\end{equation}

where the stationary amplitude \(w_c^{*}(m)\) is determined implicitly by the first-order condition

\begin{equation}
\frac{\partial\Phi}{\partial w_c}=0.
\end{equation}

At first sight, differentiating \(\Phi^{*}(m)\) with respect to the multiplicity parameter appears to require the derivative \(\frac{dw_c^{*}}{dm}\) through the chain rule. Indeed,

\begin{equation}
\frac{d\Phi^{*}}{dm}
=
\frac{\partial\Phi}{\partial m}
+
\frac{\partial\Phi}{\partial w_c}
\frac{dw_c^{*}}{dm}.
\end{equation}

However, the stationary solution satisfies

\begin{equation}
\left.
\frac{\partial\Phi}{\partial w_c}
\right|_{w_c=w_c^{*}(m)}
=0,
\end{equation}

and therefore the second term vanishes identically.

Consequently,

\begin{equation}
\boxed{
\frac{d\Phi^{*}}{dm}
=
\left.
\frac{\partial\Phi}{\partial m}
\right|_{w_c=w_c^{*}(m)}.
}
\end{equation}

This is precisely the statement of the Envelope Theorem\cite{milgrom2002envelope, samuelson1947foundations}.

The theorem shows that the variation of the optimized complexity depends only upon the explicit dependence of the reduced functional on the multiplicity parameter. The implicit variation of the stationary amplitude does not contribute.

Therefore, the determination of the optimal multiplicity requires only the explicit partial derivative of the reduced complexity with respect to \(m\).

Exactly the same argument applies to the continuous formulation. Writing

\begin{equation}
\Phi^{*}(M)
=
\Phi\!\left(\rho_c^{*}(M),M\right),
\end{equation}

the envelope theorem gives

\begin{equation}
\boxed{
\frac{d\Phi^{*}}{dM}
=
\left.
\frac{\partial\Phi}{\partial M}
\right|_{\rho_c=\rho_c^{*}(M)}.
}
\end{equation}

Hence, both the discrete and continuous optimization problems reduce to the evaluation of a single explicit derivative with respect to the multiplicity parameter.

\subsection{Envelope derivative of the optimized complexity}

After solving the stationary condition
\begin{equation}
\frac{\partial \Phi}{\partial w_c}=0,
\end{equation}
the generalized complexity becomes the one-parameter family
\begin{equation}
\Phi^{*}(m)=\Phi\!\left(w_c^{*}(m),m\right).
\end{equation}

Direct differentiation of this optimized function would ordinarily require the implicit derivative
\(
dw_c^{*}/dm
\).
However, since the stationary condition has already been imposed, the envelope theorem applies, yielding

\begin{equation}
\boxed{
\frac{d\Phi^{*}}{dm}
=
\frac{\partial \Phi}{\partial m}
}
\end{equation}

where the partial derivative is evaluated while holding \(w_c\) fixed\cite{afriat1971theory}.

Differentiating the unconstrained complexity functional with respect to the multiplicity parameter therefore gives

\begin{align}
\frac{d\Phi^{*}}{dm}
&=
(w_c-w_h)
-
w_c
\ln\!\left(\frac{w_c}{w_h}\right)
\nonumber\\
&\quad
+
\frac{
2\left[
q\,w_h^{q-1}(w_c-w_h)
-
(w_c^q-w_h^q)
\right]
}{P_q}
\nonumber\\
&\quad
+
\frac{
(w_c^{2q}-w_h^{2q})
-
2q\,w_h^{2q-1}(w_c-w_h)
}{P_{2q}},
\label{eq:raw_envelope}
\end{align}

where

\begin{equation}
P_q
=
mw_c^q+(N-m)w_h^q,
\qquad
P_{2q}
=
mw_c^{2q}+(N-m)w_h^{2q}.
\end{equation}

Equation~\eqref{eq:raw_envelope} is exact but still contains the transcendental logarithmic term. This logarithm is not independent, since the stationary condition obtained in the previous subsection satisfies

\begin{equation}
\ln\!\left(\frac{w_c}{w_h}\right)
=
-\frac{2q\left(w_c^{q-1}-w_h^{q-1}\right)}{P_q}
+
\frac{2q\left(w_c^{2q-1}-w_h^{2q-1}\right)}{P_{2q}}.
\label{eq:FOC_log}
\end{equation}

Substituting Eq.~\eqref{eq:FOC_log} into Eq.~\eqref{eq:raw_envelope}, all logarithmic terms cancel identically. After straightforward algebraic simplification, the derivative reduces to the remarkably compact form

\begin{equation}
\boxed{
\frac{d\Phi^{*}}{dm}
=
(w_c-w_h)
+
\frac{2(q-1)\left(w_c^q-w_h^q\right)}{P_q}
-
\frac{(2q-1)\left(w_c^{2q}-w_h^{2q}\right)}{P_{2q}}.
}
\label{eq:envelope_logfree}
\end{equation}

This identity is valid at every stationary point of the optimization problem and contains no transcendental terms. Consequently, the sign of the optimized complexity derivative is determined entirely by rational power functions of the stationary amplitudes.

In the following subsection, we introduce the ratio
\begin{equation}
\tau=\frac{w_c}{w_h}>1,
\end{equation}
which transforms Eq.~\eqref{eq:envelope_logfree} into a universal polynomial form whose monotonicity can be established analytically.

\subsubsection*{Continuous case}

The continuous optimization problem proceeds identically. After eliminating the normalization constraint

\[
M\rho_c+(L-M)\rho_h=1,
\]

the generalized complexity becomes the unconstrained function

\[
\Phi=\Phi(\rho_c,M),
\]

where

\[
P_q
=
M\rho_c^q+(L-M)\rho_h^q,
\qquad
P_{2q}
=
M\rho_c^{2q}+(L-M)\rho_h^{2q}.
\]

The stationary condition

\[
\frac{\partial\Phi}{\partial\rho_c}=0
\]

determines the optimal branch
\(
\rho_c=\rho_c^*(M)
\),
and therefore

\[
\Phi^*(M)
=
\Phi\!\left(\rho_c^*(M),M\right).
\]

Applying the envelope theorem yields

\[
\boxed{
\frac{d\Phi^*}{dM}
=
\frac{\partial\Phi}{\partial M},
}
\]

where the partial derivative is evaluated while holding \(\rho_c\) fixed.

Differentiating the unconstrained functional gives

\begin{align}
\frac{d\Phi^*}{dM}
&=
(\rho_c-\rho_h)
-
\rho_c
\ln\!\left(\frac{\rho_c}{\rho_h}\right)
\nonumber\\
&\quad
+
\frac{
2\left[
q\rho_h^{\,q-1}(\rho_c-\rho_h)
-
(\rho_c^q-\rho_h^q)
\right]
}{P_q}
\nonumber\\
&\quad
+
\frac{
(\rho_c^{2q}-\rho_h^{2q})
-
2q\rho_h^{\,2q-1}(\rho_c-\rho_h)
}{P_{2q}}.
\label{eq:raw_envelope_cont}
\end{align}

The continuous stationary equation satisfies

\[
\ln\!\left(\frac{\rho_c}{\rho_h}\right)
=
-
\frac{2q(\rho_c^{q-1}-\rho_h^{q-1})}{P_q}
+
\frac{2q(\rho_c^{2q-1}-\rho_h^{2q-1})}{P_{2q}},
\]

which is identical in structure to the discrete case. Substituting this relation into Eq.~\eqref{eq:raw_envelope_cont} eliminates the logarithmic term completely and yields

\[
\boxed{
\frac{d\Phi^*}{dM}
=
(\rho_c-\rho_h)
+
\frac{2(q-1)(\rho_c^q-\rho_h^q)}{P_q}
-
\frac{(2q-1)(\rho_c^{2q}-\rho_h^{2q})}{P_{2q}}.
}
\label{eq:logfree_cont}
\]

Thus, both the discrete and continuous optimization problems reduce to the same algebraic structure, differing only by the replacement

\[
(m,N,w_c,w_h)
\longrightarrow
(M,L,\rho_c,\rho_h).
\]

Consequently, the monotonicity analysis developed in the following subsection applies simultaneously to both settings.

\subsection{Reduction to a ratio formulation}

The envelope derivatives obtained in the previous subsection still involve the two stationary amplitudes explicitly. These expressions can be reduced further by introducing the ratio between the two branches.

Since the core branch always satisfies
\begin{equation}
w_c>w_h,
\end{equation}
we define
\begin{equation}
\tau=\frac{w_c}{w_h}>1.
\end{equation}

Using the normalization condition
\begin{equation}
mw_c+(N-m)w_h=1,
\end{equation}
the two amplitudes become
\begin{equation}
w_h=\frac{1}{m\tau+(N-m)},
\qquad
w_c=\frac{\tau}{m\tau+(N-m)}.
\end{equation}

Similarly, the escort moments become
\begin{align}
P_q
&=
mw_c^q+(N-m)w_h^q
=
\frac{m\tau^q+(N-m)}
{\left[m\tau+(N-m)\right]^q},
\\
P_{2q}
&=
mw_c^{2q}+(N-m)w_h^{2q}
=
\frac{m\tau^{2q}+(N-m)}
{\left[m\tau+(N-m)\right]^{2q}}.
\end{align}

For convenience, we introduce the family of rational kernels
\begin{equation}
R_k(\tau)
=
\frac{\tau^k-1}
{m\tau^k+(N-m)},
\qquad
k>0.
\label{eq:Rk_definition}
\end{equation}

Notice that
\begin{align}
R_1(\tau)
&=
\frac{w_c-w_h}{P_1},
\\
R_q(\tau)
&=
\frac{w_c^q-w_h^q}{P_q},
\\
R_{2q}(\tau)
&=
\frac{w_c^{2q}-w_h^{2q}}{P_{2q}}.
\end{align}

Substituting these identities into the log-free envelope derivative obtained previously,
\begin{equation}
\frac{d\Phi^*}{dm}
=
{w_c-w_h}
+
\frac{2(q-1)(w_c^q-w_h^q)}{P_q}
-
\frac{(2q-1)(w_c^{2q}-w_h^{2q})}{P_{2q}},
\end{equation}
and substituting the rational kernels, we obtain
\begin{equation}
\boxed{
\frac{d\Phi^*}{dm}
=
\left[
R_1(\tau)
+
2(q-1)R_q(\tau)
-
(2q-1)R_{2q}(\tau)
\right].
}
\label{eq:discrete_R_form}
\end{equation}

The continuous problem admits an identical reduction.

Defining
\begin{equation}
\tau=\frac{\rho_c}{\rho_h}>1,
\end{equation}
the normalization condition
\begin{equation}
M\rho_c+(L-M)\rho_h=1
\end{equation}
gives
\begin{equation}
\rho_h=\frac{1}{M\tau+(L-M)},
\qquad
\rho_c=\frac{\rho}{M\tau+(L-M)}.
\end{equation}

Introducing
\begin{equation}
R_k(\tau)
=
\frac{\tau^k-1}
{M\tau^k+(L-M)},
\end{equation}
the optimized continuous derivative becomes
\begin{equation}
\boxed{
\frac{d\Phi^*}{dM}
=
\left[
R_1(\tau)
+
2(q-1)R_q(\tau)
-
(2q-1)R_{2q}(\tau)
\right].
}
\label{eq:continuous_R_form}
\end{equation}

Therefore, both the discrete and continuous optimization problems reduce to the same universal algebraic structure. The proof of monotonicity is consequently reduced to establishing the sign of the common expression
\begin{equation}
R_1
+
2(q-1)R_q
-
(2q-1)R_{2q},
\end{equation}
which is carried out in the following subsection.

\subsection{Monotonicity of the universal kernel}

The optimized complexity derivative obtained in the previous subsection is governed entirely by the universal combination
\begin{equation}
\mathcal{R}(\tau)
=
R_1(\tau)
+
2(q-1)R_q(\tau)
-
(2q-1)R_{2q}(\tau),
\end{equation}
where
\begin{equation}
R_k(\tau)
=
\frac{\tau^k-1}
{m\tau^k+(N-m)},
\qquad
\tau>1.
\end{equation}

The sign of the envelope derivative therefore reduces to determining the sign of
\(
\mathcal{R}(\tau).
\)

To establish this, we first investigate the dependence of the kernel upon its exponent.

Differentiating \(R_k(\tau)\) with respect to the continuous parameter \(k\) gives
\begin{equation}
\frac{\partial R_k}{\partial k}
=
\frac{
N\tau^k\ln(\tau)
}{
\left(m\tau^k+N-m\right)^2
}.
\end{equation}

Since
\begin{equation}
N>m>0,
\qquad
\tau>1,
\qquad
\ln(\tau)>0,
\end{equation}
every factor on the right-hand side is strictly positive. Hence
\begin{equation}
\boxed{
\frac{\partial R_k}{\partial k}>0.
}
\end{equation}

Therefore the family \(R_k(\tau)\) is strictly increasing as a function of the exponent \(k\). In particular,
\begin{equation}
R_1(\tau)
<
R_q(\tau)
<
R_{2q}(\tau),
\qquad
(q>1).
\end{equation}

Next observe that
\begin{equation}
1
+
2(q-1)
=
2q-1.
\end{equation}

Hence the coefficients appearing in \(\mathcal{R}\) form a convex decomposition,
\begin{equation}
\frac{1}{2q-1}
+
\frac{2(q-1)}{2q-1}
=
1.
\end{equation}

Since both \(R_1\) and \(R_q\) are strictly smaller than \(R_{2q}\), every convex combination of \(R_1\) and \(R_q\) is likewise strictly smaller than \(R_{2q}\). Thus
\begin{equation}
\frac{1}{2q-1}R_1
+
\frac{2(q-1)}{2q-1}R_q
<
R_{2q}.
\end{equation}

Multiplying by the positive quantity \(2q-1\) immediately yields
\begin{equation}
\boxed{
R_1
+
2(q-1)R_q
-
(2q-1)R_{2q}
<
0.
}
\end{equation}

Consequently,
\begin{equation}
\boxed{
\frac{d\Phi^{*}}{dm}<0,
\qquad
\frac{d\Phi^{*}}{dM}<0.
}
\end{equation}

The optimized generalized complexity therefore decreases monotonically with increasing core multiplicity. The unique maximum is attained at the smallest admissible core measure\cite{marshall2011inequalities},
\begin{equation}
\boxed{
m=1
}
\end{equation}
for discrete systems, and
\begin{equation}
\boxed{
M\rightarrow0^{+}
}
\end{equation}
for continuous systems.

Thus, in both discrete and continuous settings, the complexity-maximizing configuration consists of a minimal high-density core embedded within an extended low-density halo\cite{bhatia1997matrix}.

\section{The Core-Halo Theorem}
\label{the core-halo theorem}
The results of the previous sections allow us to characterize the universal structure of the probability distributions maximizing the generalized statistical complexity. The following theorem summarizes the main result of the present work.

\begin{theorem}[Core--Halo Theorem for Discrete Distributions]
Let
\begin{equation}
\Phi(W) = e^{H_1(W)+(2q-2)H_q(W)-(2q-1)H_{2q}(W)}
\end{equation}
be the generalized statistical complexity defined over the probability simplex
\begin{equation}
\Delta_{N-1} = \left\{ W=(w_1,\ldots,w_N) : w_i\ge0, \quad \sum_{i=1}^{N}w_i=1 \right\},
\end{equation}
with $q>1$.

Then every global maximizer of $\Phi$ possesses exactly two distinct probability values, $w_h<w_c$, and is necessarily of the form
\begin{equation}
W^* = \left( w_c, \underbrace{w_h,\ldots,w_h}_{N-1} \right),
\end{equation}
up to permutation of the coordinates.

Equivalently, one probability component forms a distinguished core while the remaining $N-1$ components constitute a uniform halo.
\end{theorem}

\begin{proof}
Section~6 established that the stationary equation possesses exactly two positive roots, $w_h<w_c$. Consequently, every stationary distribution must have the structure
\begin{equation}
W = ( \underbrace{w_c,\ldots,w_c}_{m}, \underbrace{w_h,\ldots,w_h}_{N-m} ),
\end{equation}
for some integer $m\in\{1,\ldots,N-1\}$.

Section~7 proved that the optimized complexity decreases monotonically with increasing multiplicity:
\begin{equation}
\Phi(1) > \Phi(2) > \cdots > \Phi(N-1).
\end{equation}

Hence the maximum is attained uniquely at $m=1$. Therefore,
\begin{equation}
W^* = (w_c, w_h,\ldots,w_h),
\end{equation}
up to permutation of the coordinates. This proves the theorem.
\end{proof}

\begin{theorem}[Core--Halo Theorem for Continuous Distributions]
Let
\begin{equation}
\Phi[\rho] = e^{H_1[\rho]+(2q-2)H_q[\rho]-(2q-1)H_{2q}[\rho]}
\end{equation}
be defined over normalized probability densities satisfying
\begin{equation}
\rho(x)\ge0, \qquad \int_\Omega \rho(x)\,dx=1.
\end{equation}

Then every global maximizer possesses exactly two density levels, $\rho_h<\rho_c$, and has the form
\begin{equation}
\rho^*(x) =
\begin{cases}
\rho_c, & x\in\Omega_c, \\[1ex]
\rho_h, & x\in\Omega_h,
\end{cases}
\end{equation}
where $\Omega=\Omega_c\cup\Omega_h$, $\Omega_c\cap\Omega_h=\emptyset$, and the measure of the core region is minimal.

Hence the maximizing density consists of a single dominant core embedded within a surrounding halo.
\end{theorem}

\begin{proof}
The Euler--Lagrange equation admits exactly two density branches, $\rho_h<\rho_c$. Consequently, stationary densities are necessarily piecewise constant,
\begin{equation}
\rho(x) =
\begin{cases}
\rho_c, & x\in\Omega_c, \\
\rho_h, & x\in\Omega_h.
\end{cases}
\end{equation}

Section~7 established that the optimized complexity decreases monotonically with increasing measure of the core region, $d\Phi/dM<0$. Therefore the maximum is attained when the core measure assumes its minimum admissible value.

Hence the maximizing density consists of a single concentrated core surrounded by an extended halo. This completes the proof.
\end{proof}

\section{Universality and Structural Consequences}
\label{universality and strcutural consequences}
The core-halo theorem established in the previous section provides a complete characterization of the distributions maximizing the generalized statistical complexity. The result reveals that the extremizing distributions are neither completely ordered nor completely disordered, but instead possess an intermediate heterogeneous organization consisting of a dominant core and an extended halo.

\subsection{Emergence of heterogeneity}

The maximization problem initially allows arbitrary probability distributions over the simplex. A priori, one might expect the maximizing configurations to exhibit complicated structures involving many different probability values. However, the variational analysis shows that the stationary equation admits exactly two positive branches, $w_h<w_c$, and the multiplicity analysis further demonstrates that the global maximum is attained when the larger branch appears exactly once.

Consequently, the complexity-maximizing distribution possesses the universal structure
\begin{equation}
W^* = (w_c,w_h,\ldots,w_h),
\end{equation}
up to permutation of the coordinates. Thus, maximal complexity emerges from the coexistence of two distinct scales rather than from homogeneous randomness.

\subsection{Exclusion of complete order and complete disorder}

The two extremal configurations of the simplex correspond to:
\begin{enumerate}
    \item The completely uniform distribution: $W_u = \left(\frac{1}{N},\ldots,\frac{1}{N}\right)$,
    \item The completely localized Dirac state: $W_d = (1,0,\ldots,0)$.
\end{enumerate}

Both configurations satisfy $\Phi(W_u)=\Phi(W_d)=0$. Therefore, neither perfect disorder nor perfect order maximizes complexity. 

The complexity maximum necessarily lies between these two extremes. In this sense, complexity represents a compromise between concentration and diversity.

\subsection{Universality with respect to dimension}

An important feature of the theorem is that the resulting structure is independent of the dimension $N$ of the simplex. Although the exact numerical values of $w_c$ and $w_h$ depend on $N$ and on the deformation parameter $q$, the qualitative structure remains invariant:
\begin{equation}
\boxed{\text{one core + one halo}.}
\end{equation}

Thus, the theorem reveals a universal geometrical principle underlying the generalized statistical complexity.

\subsection{Continuous counterpart}

Exactly the same mechanism persists for continuous probability densities. The variational equation generates two density branches, $\rho_h<\rho_c$, and the complexity maximum is obtained when the measure of the high-density region assumes its smallest admissible value.

An consequence of this optimization is that the maximizing density takes the form
\begin{equation}
\rho^*(x) =
\begin{cases}
\rho_c, & x\in\Omega_c, \\
\rho_h, & x\in\Omega_h,
\end{cases}
\end{equation}
where the core region $\Omega_c$ occupies a minimal measure. Hence, the core--halo principle survives the passage from discrete probability vectors to continuous probability densities.

\subsection{Interpretation}

The theorem suggests that maximal complexity is generated by the coexistence of localized concentration and a distributed background. Symbolically, this can be framed as:
\begin{equation}
\text{Complexity} = \text{Order} \times \text{Disorder},
\end{equation}
demonstrating that neither complete randomness nor complete localization is sufficient to maximize complexity. 

Instead, the optimum configuration is achieved when these two competing tendencies coexist simultaneously. The appearance of the core--halo geometry may therefore be regarded as a universal manifestation of the perfect balance between concentration and diversity.

\section{Numerical Illustrations and Examples}
\label{numerical illustrations and examples}
The analytical results obtained in the previous sections establish the existence and uniqueness of the core--halo structure maximizing the generalized statistical complexity. In order to visualize these theoretical predictions and verify their consistency, we now present several numerical examples.

The purpose of this section is twofold. First, it provides a direct illustration of the solutions of the stationary equations and their dependence upon the deformation parameter $q$. Second, it demonstrates the monotonic decrease of the optimized complexity with increasing multiplicity, thereby confirming the $m=1$ theorem.

\subsection{Numerical procedure}

For fixed values of the dimension $N$ and deformation parameter $q$, the stationary probabilities are determined from the coupled equations
\begin{equation}
mw_c+(N-m)w_h=1,
\end{equation}
together with the equality of the stationary equations
\begin{equation}
\ln w_c + \frac{2q}{P_q}w_c^{q-1} - \frac{2q}{P_{2q}}w_c^{2q-1} = \ln w_h + \frac{2q}{P_q}w_h^{q-1} - \frac{2q}{P_{2q}}w_h^{2q-1},
\end{equation}
where
\begin{equation}
P_q = mw_c^q+(N-m)w_h^q,
\end{equation}
and
\begin{equation}
P_{2q} = mw_c^{2q}+(N-m)w_h^{2q}.
\end{equation}

For each admissible multiplicity $m=1,2,\ldots,N-1$, the nonlinear system is solved numerically to obtain the corresponding values of $w_c$ and $w_h$.

The optimized complexity is then computed from
\begin{equation}
\Phi_m = \exp \left[ H_1 + (2q-2)H_q - (2q-1)H_{2q} \right].
\end{equation}

The same procedure may be extended to continuous probability densities by replacing the discrete multiplicity $m$ by the measure $M$ of the core region and solving the corresponding Euler--Lagrange equations.

The numerical results obtained from these equations allow us to investigate:
\begin{enumerate}
\item the dependence of the stationary probabilities $(w_c,w_h)$ on the deformation parameter $q$;
\item the evolution of the optimized complexity as a function of the multiplicity parameter $m$;
\item the convergence towards the uniform and Dirac limits;
\item the emergence of the core--halo geometry;
\item the validity of the monotonicity theorem, $\Phi(1) > \Phi(2) > \cdots > \Phi(N-1)$.
\end{enumerate}

In the following subsections, several representative examples are presented in order to illustrate these properties and to compare the theoretical predictions with numerical computations.

\subsection{Dependence of the stationary branches on the deformation parameter}

One of the most important characteristics of the generalized statistical complexity is the dependence of the stationary branches upon the deformation parameter $q$. Since the quantities $w_c$ and $w_h$ are determined implicitly by the stationary equations, their values vary continuously with the parameter $q$.

For a fixed dimension $N$ and fixed multiplicity $m$, the stationary probabilities are obtained from
\begin{equation}
mw_c+(N-m)w_h=1,
\end{equation}
together with
\begin{equation}
\ln\frac{w_c}{w_h} = -\frac{2q}{P_q} \left( w_c^{q-1}-w_h^{q-1} \right) + \frac{2q}{P_{2q}} \left( w_c^{2q-1}-w_h^{2q-1} \right),
\end{equation}
where
\begin{equation}
P_q = mw_c^q+(N-m)w_h^q,
\end{equation}
and
\begin{equation}
P_{2q} = mw_c^{2q}+(N-m)w_h^{2q}.
\end{equation}

These equations define two stationary branches, $w_h(q)<w_c(q)$, whose separation depends on the value of the deformation parameter.

For values of $q$ close to unity, the difference between the two branches remains relatively small and the maximizing distribution approaches a nearly homogeneous state. In this regime, the generalized complexity differs only slightly from the ordinary LMC complexity.

As $q$ increases, the contribution of the higher-order Rényi entropies becomes progressively stronger. Consequently, the competition between disorder and disequilibrium becomes more pronounced, leading to a larger separation between the two stationary branches. The dominant probability $w_c$ increases, while the halo probability $w_h$ decreases.

Hence, increasing $q$ enhances the degree of heterogeneity of the complexity-maximizing distribution. This behavior is illustrated schematically by $w_c(q)\uparrow$ and $w_h(q)\downarrow$ as $q\uparrow$.

The resulting widening of the gap $\Delta(q) = w_c(q)-w_h(q)$ may therefore be regarded as a quantitative measure of the emergence of the core--halo structure.

In the limit of large $q$, the dominant branch becomes increasingly concentrated, whereas the halo branch becomes progressively smaller. Nevertheless, the two branches remain finite and positive, thereby preserving the coexistence of concentration and diversity required for maximal complexity.

Therefore, the deformation parameter $q$ acts as a control parameter governing the strength of the core--halo separation.

Numerical solutions of the stationary equations confirm these predictions and show that the functions $w_c(q)$ and $w_h(q)$ vary smoothly with $q$ and exhibit monotonic tendencies over a broad range of parameter values.

The dependence of the branch separation on $q$ demonstrates that the generalized complexity generates an entire family of core--halo distributions, with the ordinary LMC complexity appearing as a particular member of this family.

\subsection{Complexity as a function of multiplicity}

The analytical developments of Section~7 established that the generalized statistical complexity is a strictly decreasing function of the multiplicity parameter. In particular, the $m$-core theorem predicts that $\Phi(1) > \Phi(2) > \cdots > \Phi(N-1)$, and therefore that the global maximum is always attained by the one-core configuration.

The purpose of this subsection is to investigate this behavior numerically and to illustrate the monotonic dependence of the optimized complexity upon the multiplicity parameter.

For fixed values of the dimension $N$ and deformation parameter $q$, the stationary equations are solved for each admissible multiplicity $m=1,2,\ldots,N-1$, yielding the corresponding branches $w_c(m)$ and $w_h(m)$.

The optimized complexity associated with the $m$-core configuration is then given by
\begin{equation}
\Phi(m) = \exp \left[ H_1 + (2q-2)H_q - (2q-1)H_{2q} \right],
\end{equation}
where
\begin{align}
H_1 &= -mw_c\ln w_c - (N-m)w_h\ln w_h, \\
H_q &= \frac{1}{1-q} \ln \left( mw_c^q+(N-m)w_h^q \right), \\
H_{2q} &= \frac{1}{1-2q} \ln \left( mw_c^{2q} + (N-m)w_h^{2q} \right).
\end{align}

The numerical solutions reveal a remarkable and highly regular behavior. For every value of $q$ examined, the complexity decreases monotonically with increasing multiplicity. Consequently, the sequence $\Phi(1),\Phi(2),\ldots,\Phi(N-1)$ forms a strictly decreasing hierarchy.

This hierarchy may be represented schematically as $\Phi(1) > \Phi(2) > \Phi(3) > \cdots > \Phi(N-1)$. Therefore, the one-core configuration possesses the largest possible complexity, whereas configurations containing several core components correspond to progressively smaller complexity values.

The physical interpretation of this result is straightforward. Increasing the multiplicity spreads the concentrated region over a larger number of coordinates and therefore reduces the asymmetry between the core and the halo. As the distinction between the two branches becomes weaker, the system moves closer to homogeneous configurations and consequently loses complexity. Hence, the complexity decreases as the number of core components increases.

In particular, the limiting configuration $m=N-1$ corresponds to a situation in which almost all components belong to the dominant branch. In this regime, the contrast between the two branches becomes weak and the resulting complexity is considerably smaller than that of the one-core state.

The numerical computations therefore provide a direct verification of the analytical monotonicity theorem established in Section~7. They demonstrate that the inequality
\begin{equation}
\boxed{\Phi(1) > \Phi(2) > \cdots > \Phi(N-1)}
\end{equation}
holds universally and independently of the dimension $N$ and of the deformation parameter $q$.

Consequently, the numerical analysis confirms that the one-core configuration constitutes the unique global maximizer of the generalized statistical complexity.

For continuous probability densities, an entirely analogous behavior is observed. Replacing the discrete multiplicity $m$ by the measure $M$ of the core region, one finds that $d\Phi/dM < 0$, showing that the optimized complexity decreases continuously as the measure of the core region increases.

Therefore, both discrete and continuous systems exhibit the same universal monotonic structure, providing strong numerical support for the generalized core--halo theorem.

\subsection{Visualization of the core--halo distributions}

The numerical analysis presented in the previous subsection demonstrates that the one-core configuration maximizes the generalized statistical complexity. We now examine the geometrical structure of these maximizing distributions.

For discrete systems, the maximizing probability vector possesses the universal form
\begin{equation}
W^* = \left( w_c, \underbrace{ w_h,\ldots,w_h }_{N-1} \right),
\end{equation}
where $w_c>w_h>0$.

Thus, a single component carries a larger probability weight, whereas the remaining components form a nearly uniform background. This structure is naturally interpreted as a core surrounded by a halo.

From a geometrical viewpoint, the probability vector may therefore be regarded as consisting of two distinct scales:
\begin{enumerate}
\item a localized high-probability region represented by the core probability $w_c$;
\item an extended low-probability background represented by the halo probability $w_h$.
\end{enumerate}

The separation between these two scales is measured by the gap $\Delta = w_c-w_h$.

For values of $q$ close to unity, the branch separation remains relatively small, and the probability vector appears almost homogeneous. As the deformation parameter increases, the dominant component becomes progressively larger while the halo components become smaller. Consequently, the contrast between the two scales becomes increasingly pronounced.

Schematically, one may visualize the evolution as $\left(\frac{1}{N},\ldots,\frac{1}{N}\right) \longrightarrow \left(w_c,w_h,\ldots,w_h\right)$, where the difference between $w_c$ and $w_h$ grows with increasing $q$.

The corresponding probability histograms exhibit a characteristic two-level structure. One observes a single dominant peak associated with the core probability together with a broad plateau formed by the halo probabilities. This behavior is illustrated conceptually in Fig.~\ref{fig:core_halo_discrete}.

\begin{figure}[h]
\centering
\includegraphics[width=0.8\textwidth]{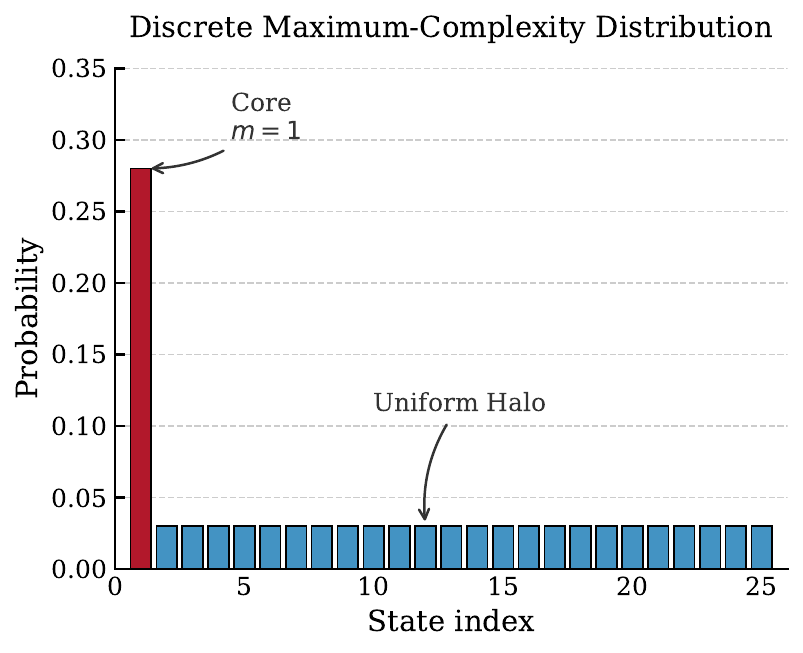}
\caption{Schematic representation of the complexity-maximizing discrete distribution. One dominant component forms the core, while the remaining components constitute a uniform halo.}
\label{fig:core_halo_discrete}
\end{figure}

For continuous probability densities, the same phenomenon persists. The maximizing density assumes the form
\begin{equation}
\rho^*(x) =
\begin{cases}
\rho_c, & x\in\Omega_c, \\[1ex]
\rho_h, & x\in\Omega_h,
\end{cases}
\end{equation}
with $\rho_c>\rho_h$.

Therefore, the density profile consists of a compact region of high density immersed within a low-density background. In analogy with the discrete case, these regions are naturally interpreted as the core and the halo, respectively.

The measure of the core region is minimal, whereas the halo occupies the majority of the support. This structure is illustrated schematically in Fig.~\ref{fig:core_halo_continuous}.

\begin{figure}[h]
\centering
\includegraphics[width=0.8\textwidth]{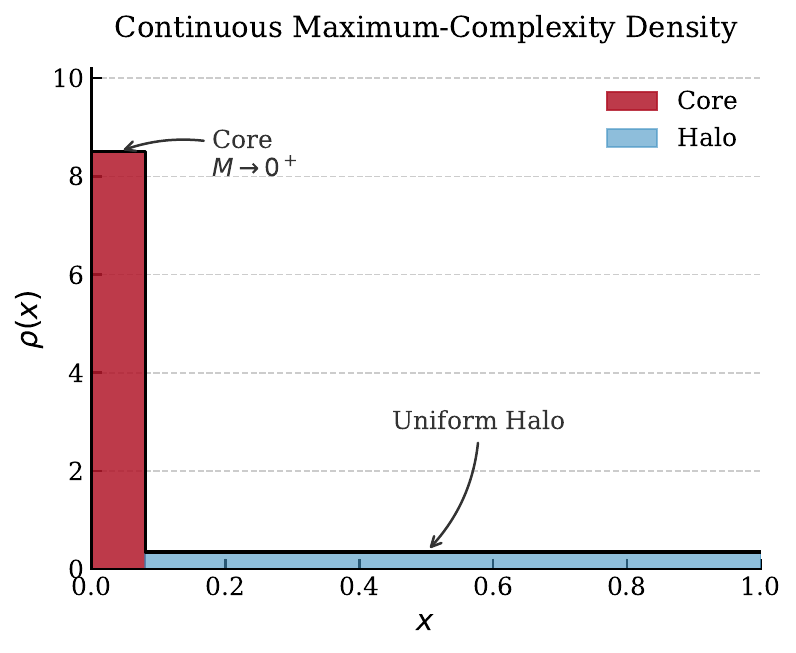}
\caption{Schematic representation of the complexity-maximizing continuous density. A localized high-density core is embedded within an extended low-density halo.}
\label{fig:core_halo_continuous}
\end{figure}

Hence, despite the fundamental differences between discrete probability vectors and continuous probability densities, both systems exhibit exactly the same two-scale organization. The emergence of this structure constitutes one of the most remarkable manifestations of the generalized statistical complexity and provides strong numerical support for the core--halo theorem established in Section~8.

\subsection{Comparison with the ordinary LMC complexity}

The generalized statistical complexity introduced in the present work contains the ordinary LMC complexity as a particular member of a much larger family. It is therefore natural to investigate the relation between the two theories and to understand how the deformation parameter modifies the structure of the complexity-maximizing distributions.

The original LMC complexity is defined as
\begin{equation}
C_{\mathrm{LMC}} = H\times D,
\end{equation}
where
\begin{equation}
H = -\sum_{i=1}^{N}w_i\ln w_i
\end{equation}
denotes the Shannon entropy and
\begin{equation}
D = \sum_{i=1}^{N} \left( w_i-\frac{1}{N} \right)^2
\end{equation}
represents the disequilibrium.

In contrast, the generalized complexity investigated in the present work takes the form
\begin{equation}
\Phi = \exp \left[ H_1 + (2q-2)H_q - (2q-1)H_{2q} \right],
\end{equation}
which depends simultaneously upon three Rényi entropies.

Consequently, whereas the ordinary LMC complexity balances disorder against Euclidean disequilibrium, the generalized complexity measures the competition among information contained at three different scales.

An important difference emerges from the structure of the maximizing distributions.

For the ordinary LMC complexity, the complexity maximum is already known to possess a core--halo form, $W^* = (w_c,w_h,\ldots,w_h)$, with one dominant probability component and a uniform halo.

The generalized complexity preserves exactly the same qualitative structure. However, the deformation parameter $q$ introduces a continuous family of complexity measures and consequently generates an entire family of core--halo distributions. Therefore, the ordinary LMC complexity appears as a particular member of a broader universality class.

From the numerical analysis, one observes that increasing $q$ tends to enlarge the difference $\Delta = w_c-w_h$, thereby strengthening the distinction between the core and the halo.

Hence the parameter $q$ acts as a measure of heterogeneity. Small values of $q$ produce relatively weak contrasts between the two branches, whereas larger values lead to increasingly pronounced core--halo structures. Schematically, one finds $q\uparrow \implies \Delta\uparrow$.

Thus the generalized theory interpolates between different degrees of structural organization while preserving the same underlying geometry.

Another remarkable feature is the persistence of the $m=1$ theorem. Numerical computations indicate that the hierarchy $\Phi(1) > \Phi(2) > \cdots > \Phi(N-1)$ holds throughout the entire range of deformation parameters examined.

Therefore, the one-core structure is not a special property of the ordinary LMC complexity but rather a universal consequence of the generalized variational equations. This observation suggests that the core--halo geometry is considerably more fundamental than the particular choice of complexity functional.

Indeed, both the ordinary and generalized theories reveal the same qualitative principle:

\begin{center}
\emph{Maximum complexity does not arise from complete order or complete disorder, but from the coexistence of concentration and diversity.}
\end{center}

Consequently, the generalized statistical complexity should be viewed not as a replacement for the ordinary LMC complexity, but as its natural extension. The original LMC measure corresponds merely to one point inside a continuous family of complexity measures sharing the same universal geometrical structure.

\section{Physical Interpretation and Universality}
\label{physical interpretation and universality}
The preceding sections establish that the generalized statistical complexity possesses a remarkable and highly rigid structure. Regardless of the dimension of the probability space, the value of the deformation parameter, or whether the system is discrete or continuous, the maximizing configurations always consist of two distinct probability branches organized into a core--halo geometry.

This universality suggests that the phenomenon transcends the specific form of the complexity measure and reflects a more fundamental principle governing the interplay between order and disorder.

\subsection{Complexity as a balance between concentration and diversity}

Two limiting configurations naturally exist in every probability space. 

The first corresponds to complete disorder, namely the uniform distribution $W_u= \left(\frac{1}{N},\ldots,\frac{1}{N}\right)$, which maximizes entropy while minimizing concentration. The second corresponds to complete order, represented by the Dirac distribution $W_d= (1,0,\ldots,0)$, which maximizes concentration while minimizing entropy.

Both extremes possess vanishing complexity, $\Phi(W_u)=\Phi(W_d)=0$. Therefore, neither complete randomness nor complete localization is capable of generating maximal complexity. Instead, the maximum necessarily occurs between these two extremes.

Hence, complexity emerges from the coexistence of two competing tendencies:
\begin{enumerate}
\item \emph{Diversity}, which favors entropy and tends to spread probability uniformly;
\item \emph{Concentration}, which favors disequilibrium and tends to localize probability.
\end{enumerate}

The core--halo structure represents the equilibrium established between these two antagonistic mechanisms.

\subsection{Emergence of two scales}

The variational analysis revealed that the stationary equation admits exactly two positive solutions, $w_h<w_c$. Consequently, the maximizing distribution contains two distinct probability scales. 

The larger branch $w_c$ defines a localized core, while the smaller branch $w_h$ forms an extended halo. Thus, maximal complexity is associated with the emergence of a two-scale organization.

This behavior is reminiscent of numerous physical systems in which highly concentrated regions coexist with diffuse backgrounds. Examples include:
\begin{itemize}
\item Galaxies embedded within dark matter halos,
\item Stars surrounded by atmospheres,
\item Atomic nuclei immersed in electron clouds,
\item Condensed phases coexisting with dilute phases,
\item Coherent structures emerging from turbulent backgrounds,
\item Cluster formation in complex networks.
\end{itemize}

Although the present theory is purely information-theoretic, the appearance of the same geometrical organization suggests the existence of deep connections between complexity and self-organization.

\subsection{Universality across discrete and continuous systems}

One of the most striking results obtained in this work is the complete equivalence between discrete and continuous probability spaces. 

In the discrete case, the maximizer assumes the form $W^* = (w_c,w_h,\ldots,w_h)$, whereas in the continuous case the maximizing density becomes
\begin{equation}
\rho^*(x) =
\begin{cases}
\rho_c, & x\in\Omega_c, \\
\rho_h, & x\in\Omega_h.
\end{cases}
\end{equation}

Although the mathematical representations differ, the underlying structure remains identical. In both cases the system decomposes into:
\begin{equation}
\boxed{\text{Core} + \text{Halo}.}
\end{equation}

Therefore, the core--halo principle appears to constitute a universal geometrical property of complexity maximization.

\subsection{A universal principle of maximal complexity}

The results obtained in this work suggest the following principle:
\begin{center}
\emph{Maximum complexity arises neither from perfect order nor from complete disorder, but from the coexistence of concentration and diversity.}
\end{center}

Equivalently,
\begin{center}
\emph{The most complex states are characterized by the simultaneous presence of localized structures and extended backgrounds.}
\end{center}

From this viewpoint, the core--halo theorem may be regarded as a mathematical manifestation of the balance between order and randomness. This principle appears to be independent of the particular dimension of the system, the nature of the probability space, and even the specific form of the generalized statistical complexity.

Consequently, the emergence of the core--halo geometry may represent a universal signature of complexity itself.

\section{Conclusions and Future Perspectives}
\label{conclusion and future perspective}
In this work, we have developed a generalized formulation of the López--Ruiz, Mancini, and Calbet statistical complexity based entirely on Rényi entropies. The resulting complexity functional
\begin{equation}
\Phi = \exp \left[ H_1 + (2q-2)H_q - (2q-1)H_{2q} \right]
\end{equation}
extends the ordinary LMC complexity and introduces a continuous family of complexity measures parameterized by the deformation parameter $q$.

The variational analysis of this functional revealed several remarkable features. First, the stationary equation was shown to possess at most two positive solutions. By combining this result with the observation that both the completely ordered and completely disordered configurations yield zero complexity, we established that the possibilities of zero and one stationary roots are excluded. Consequently, the stationary equation necessarily possesses exactly two positive branches.

This result immediately leads to the emergence of a two-valued structure characterized by a dominant branch and a secondary branch. In the discrete setting, the maximizing probability distribution assumes the form
\begin{equation}
W = ( \underbrace{w_c,\ldots,w_c}_{m}, \underbrace{w_h,\ldots,w_h}_{N-m} ),
\end{equation}
where $w_c>w_h$.

The continuous case exhibits an analogous behavior, with the maximizing density taking two constant values over complementary regions of the support. Thus, both discrete and continuous probability spaces display the same fundamental geometrical organization.

The multiplicity analysis performed through the envelope theorem further demonstrated that the optimized complexity decreases monotonically with the multiplicity parameter. Consequently, $\Phi(1) > \Phi(2) > \cdots > \Phi(N-1)$, showing that the one-core configuration is the unique global maximizer.

These results culminated in the discrete and continuous core--halo theorems. The complexity-maximizing distributions were shown to possess a universal structure consisting of a localized core embedded within an extended halo. This geometry emerges naturally from the competition between entropy and concentration and is independent of the dimension of the probability space or the value of the deformation parameter.

An important conceptual consequence of the present analysis is that maximal complexity is achieved neither by complete randomness nor by complete order. The uniform distribution and the Dirac distribution both possess vanishing complexity. Instead, complexity reaches its maximum through the coexistence of concentration and diversity. From this perspective, the core--halo geometry may be interpreted as a universal manifestation of the balance between order and disorder.

The results obtained here suggest that the core--halo principle is considerably more fundamental than the specific generalized complexity functional considered in this work. Indeed, the emergence of two scales and the preference for one-core configurations indicate the existence of a deeper geometrical mechanism underlying complexity maximization.

Several directions for future investigations naturally arise.

First, it would be interesting to extend the present analysis to other generalized entropic frameworks, including Tsallis, Sharma--Mittal, and Kaniadakis entropies, in order to determine whether the core--halo theorem persists beyond the Rényi family.

Second, the relation between the present theory and majorization could provide a deeper understanding of the ordering properties of complexity-maximizing distributions. In particular, it would be worthwhile to investigate whether the one-core configuration possesses an extremal characterization in terms of majorization lattices.

Third, one may ask whether analogous structures appear in quantum information theory. The extension of the present framework to density operators and quantum Rényi entropies may lead to the existence of quantum core--halo states.

Another promising direction concerns dynamical systems. The present work addresses static maximization problems, but it remains an open question whether nonequilibrium processes naturally evolve toward core--halo configurations. Such an investigation may reveal connections between complexity optimization and self-organization.

Finally, the remarkable similarity between the mathematical structures obtained here and the core--halo patterns observed in many physical systems suggests that the present results may have applications beyond information theory. Examples include gravitational systems, network theory, condensed matter physics, and biological organization.

In summary, the present work provides a complete analytical characterization of the distributions maximizing the generalized statistical complexity. The resulting theory reveals a universal geometrical principle: maximal complexity emerges from the coexistence of a concentrated core and an extended halo. This suggests that the core--halo structure may constitute a fundamental signature of complexity itself.

\section{Data Availability}
No new experimental data were created or analyzed in this study. All mathematical proofs and analytical results are contained within the manuscript.

\section{Declaration of AI Usage}
During the preparation of this manuscript, the author utilized generative AI tools solely for the purposes of language polishing, grammar checking, and manuscript formatting. These tools were not used to generate any scientific concepts, mathematical proofs, or primary results. The author has thoroughly reviewed and edited all AI-assisted text and assumes full responsibility for the final content and integrity of this work.

\appendix
\label{Appendix}
\section{Detailed derivation of the Euler--Lagrange equations}

In this appendix we present the complete derivation of the stationary equations underlying the generalized statistical complexity introduced in the main text.

\subsection{Discrete case}

Consider a discrete probability distribution
\begin{equation}
W=(w_1,w_2,\ldots,w_N),
\end{equation}
subject to the normalization constraint
\begin{equation}
\sum_{i=1}^{N}w_i=1.
\end{equation}

Define
\begin{equation}
P_q=\sum_{i=1}^{N}w_i^q, \qquad P_{2q}=\sum_{i=1}^{N}w_i^{2q},
\end{equation}
and let
\begin{equation}
H_1(W) = -\sum_{i=1}^{N}w_i\ln w_i
\end{equation}
denote the Shannon entropy.

The generalized complexity functional considered in the main text is
\begin{equation}
\Phi = H_1(W) + \ln P_{2q} - 2\ln P_q.
\end{equation}

Equivalently, using the Rényi entropy representation,
\begin{equation}
\Phi = H_1 + (2q-2)H_q - (2q-1)H_{2q}.
\end{equation}

Introducing a Lagrange multiplier $\lambda$ associated with the normalization constraint, we define the Lagrangian
\begin{equation}
\mathcal{L} = H_1 + \ln P_{2q} - 2\ln P_q + \lambda \left( \sum_{i=1}^{N}w_i-1 \right).
\end{equation}

Differentiating the Shannon entropy gives
\begin{equation}
\frac{\partial H_1}{\partial w_i} = -(\ln w_i+1).
\end{equation}

Furthermore,
\begin{equation}
\frac{\partial P_q}{\partial w_i} = qw_i^{q-1},
\end{equation}
and
\begin{equation}
\frac{\partial P_{2q}}{\partial w_i} = 2qw_i^{2q-1}.
\end{equation}

Hence,
\begin{equation}
\frac{\partial}{\partial w_i}\ln P_{2q} = \frac{2qw_i^{2q-1}}{P_{2q}},
\end{equation}
while
\begin{equation}
\frac{\partial}{\partial w_i}\left(-2\ln P_q\right) = -\frac{2qw_i^{q-1}}{P_q}.
\end{equation}

Therefore,
\begin{equation}
\frac{\partial\mathcal{L}}{\partial w_i} = -(\ln w_i+1) + \frac{2qw_i^{2q-1}}{P_{2q}} - \frac{2qw_i^{q-1}}{P_q} + \lambda.
\end{equation}

Imposing the stationarity condition
\begin{equation}
\frac{\partial\mathcal{L}}{\partial w_i}=0,
\end{equation}
one obtains
\begin{equation}
-(\ln w_i+1) + \frac{2qw_i^{2q-1}}{P_{2q}} - \frac{2qw_i^{q-1}}{P_q} + \lambda = 0.
\end{equation}

Multiplying by $-1$ yields
\begin{equation}
\ln w_i + \frac{2q}{P_q}w_i^{q-1} - \frac{2q}{P_{2q}}w_i^{2q-1} + (1-\lambda) = 0.
\end{equation}

Introducing the abbreviations
\begin{equation}
A_q = \frac{2q}{P_q},
\end{equation}
\begin{equation}
B_q = \frac{2q}{P_{2q}},
\end{equation}
and
\begin{equation}
C = 1-\lambda,
\end{equation}
the stationary equation assumes the compact form
\begin{equation}
\boxed{\ln w_i + A_qw_i^{q-1} - B_qw_i^{2q-1} + C = 0.}
\end{equation}

This equation constitutes the fundamental stationary equation studied throughout the paper.

\subsection{Continuous case}

Consider a probability density function
\begin{equation}
\rho(x)\ge0,
\end{equation}
defined on a domain $\Omega$, satisfying the normalization condition
\begin{equation}
\int_{\Omega}\rho(x)\,dx=1.
\end{equation}

Define
\begin{equation}
P_q=\int_{\Omega}\rho(x)^q\,dx, \qquad P_{2q}=\int_{\Omega}\rho(x)^{2q}\,dx,
\end{equation}
and let
\begin{equation}
H_1[\rho] = -\int_{\Omega}\rho(x)\ln\rho(x)\,dx
\end{equation}
denote the continuous Shannon entropy.

The generalized complexity functional is
\begin{equation}
\Phi[\rho] = H_1[\rho] + \ln P_{2q} - 2\ln P_q.
\end{equation}

Equivalently, in terms of Rényi entropies,
\begin{equation}
\Phi = H_1 + (2q-2)H_q - (2q-1)H_{2q}.
\end{equation}

Introducing a Lagrange multiplier $\lambda$ associated with normalization, we define
\begin{equation}
\mathcal{L}[\rho] = H_1[\rho] + \ln P_{2q} - 2\ln P_q + \lambda \left( \int_{\Omega}\rho(x)\,dx-1 \right).
\end{equation}

To obtain the stationary density, we compute the functional derivative.

The Shannon contribution yields
\begin{equation}
\frac{\delta H_1}{\delta\rho(x)} = -(\ln\rho(x)+1).
\end{equation}

Furthermore,
\begin{equation}
\frac{\delta P_q}{\delta\rho(x)} = q\rho(x)^{q-1},
\end{equation}
and
\begin{equation}
\frac{\delta P_{2q}}{\delta\rho(x)} = 2q\rho(x)^{2q-1}.
\end{equation}

Therefore,
\begin{equation}
\frac{\delta}{\delta\rho(x)} \ln P_{2q} = \frac{2q\rho(x)^{2q-1}}{P_{2q}},
\end{equation}
while
\begin{equation}
\frac{\delta}{\delta\rho(x)} \left(-2\ln P_q\right) = -\frac{2q\rho(x)^{q-1}}{P_q}.
\end{equation}

Hence the Euler--Lagrange equation becomes
\begin{equation}
\frac{\delta\mathcal{L}}{\delta\rho(x)} = -(\ln\rho(x)+1) + \frac{2q\rho(x)^{2q-1}}{P_{2q}} - \frac{2q\rho(x)^{q-1}}{P_q} + \lambda.
\end{equation}

Imposing stationarity,
\begin{equation}
\frac{\delta\mathcal{L}}{\delta\rho(x)} = 0,
\end{equation}
gives
\begin{equation}
-(\ln\rho(x)+1) + \frac{2q\rho(x)^{2q-1}}{P_{2q}} - \frac{2q\rho(x)^{q-1}}{P_q} + \lambda = 0.
\end{equation}

Multiplying by $-1$, we obtain
\begin{equation}
\ln\rho(x) + \frac{2q}{P_q}\rho(x)^{q-1} - \frac{2q}{P_{2q}}\rho(x)^{2q-1} + (1-\lambda) = 0.
\end{equation}

Introducing
\begin{equation}
A_q = \frac{2q}{P_q},
\end{equation}
\begin{equation}
B_q = \frac{2q}{P_{2q}},
\end{equation}
and
\begin{equation}
C = 1-\lambda,
\end{equation}
the stationary equation assumes the compact form
\begin{equation}
\boxed{\ln\rho(x) + A_q\rho(x)^{q-1} - B_q\rho(x)^{2q-1} + C = 0.}
\end{equation}

This equation is identical in structure to the discrete stationary equation and forms the basis of the continuous core--halo analysis developed in the main text.

\section{Shape analysis of the stationary equation}

\subsection{The stationary function}

The stationary equations derived in Appendix A may be written in the unified form
\begin{equation}
f(x) = \ln x + A_qx^{q-1} - B_qx^{2q-1} + C,
\end{equation}
where
\begin{equation}
A_q = \frac{2q}{P_q}, \qquad B_q = \frac{2q}{P_{2q}},
\end{equation}
and $C=1-\lambda$.

The stationary solutions correspond to the positive roots of
\begin{equation}
f(x)=0.
\end{equation}

Since
\begin{equation}
\lim_{x\to0^+}\ln x=-\infty,
\end{equation}
one has
\begin{equation}
\lim_{x\to0^+}f(x)=-\infty.
\end{equation}

Similarly, because the highest-order contribution is negative,
\begin{equation}
-B_qx^{2q-1},
\end{equation}
it follows that
\begin{equation}
\lim_{x\to\infty}f(x)=-\infty.
\end{equation}

Hence the stationary function tends to negative infinity at both ends of the positive real axis. Consequently, the number of positive roots is completely determined by the number and arrangement of its critical points.

To characterize the extrema of the stationary equation, we introduce the auxiliary function obtained from the first derivative.

\subsection{Critical points of the stationary equation}

The extrema of the stationary function
\begin{equation}
f(x) = \ln x + A_qx^{q-1} - B_qx^{2q-1} + C
\end{equation}
are determined by its first derivative.

Differentiating with respect to $x$ yields
\begin{equation}
f'(x) = \frac{1}{x} + (q-1)A_qx^{q-2} - (2q-1)B_qx^{2q-2}.
\end{equation}

Since $x>0$, multiplication by $x$ preserves the location of the zeros. It is therefore convenient to introduce the auxiliary function
\begin{equation}
\boxed{g(x) = xf'(x) = 1 + (q-1)A_qx^{q-1} - (2q-1)B_qx^{2q-1}.}
\end{equation}

Consequently,
\begin{equation}
f'(x)=0 \qquad\Longleftrightarrow\qquad g(x)=0.
\end{equation}

Hence the critical points of the stationary equation are in one-to-one correspondence with the positive roots of the auxiliary function $g(x)$.

To investigate the structure of $g(x)$, we differentiate once more:
\begin{equation}
g'(x) = (q-1)^2A_qx^{q-2} - (2q-1)^2B_qx^{2q-2}.
\end{equation}

Setting
\begin{equation}
g'(x)=0
\end{equation}
gives
\begin{equation}
(q-1)^2A_qx^{q-2} = (2q-1)^2B_qx^{2q-2},
\end{equation}
or equivalently,
\begin{equation}
x^q = \frac{(q-1)^2A_q}{(2q-1)^2B_q}.
\end{equation}

Therefore the derivative $g'(x)$ possesses the unique positive zero
\begin{equation}
\boxed{x_c = \left( \frac{(q-1)^2A_q}{(2q-1)^2B_q} \right)^{1/q}.}
\end{equation}

Thus the auxiliary function $g(x)$ possesses exactly one critical point. Moreover,

\begin{equation}
g'(x)>0, \qquad 0<x<x_c,
\end{equation}

and

\begin{equation}
g'(x)<0, \qquad x>x_c.
\end{equation}

Hence $g(x)$ increases monotonically on $(0,x_c)$ and decreases monotonically on $(x_c,\infty)$, so that $x_c$ corresponds to the unique global maximum of the auxiliary function.

The existence and uniqueness of this maximum constitute the key geometrical property governing the number of stationary branches of the complexity equation.

\subsection{Uniqueness of the positive root of the auxiliary function}

The preceding subsection established that the auxiliary function
\begin{equation}
g(x) = 1 + (q-1)A_qx^{q-1} - (2q-1)B_qx^{2q-1}
\end{equation}
possesses a unique critical point located at
\begin{equation}
x_c = \left( \frac{(q-1)^2A_q}{(2q-1)^2B_q} \right)^{1/q},
\end{equation}
which corresponds to its global maximum.

To determine the number of positive roots of $g(x)$, we first examine its behavior near the origin. Since
\begin{equation}
\lim_{x\to0^+}x^{q-1}=0, \qquad \lim_{x\to0^+}x^{2q-1}=0,
\end{equation}
one immediately obtains
\begin{equation}
\boxed{g(0^+)=1.}
\end{equation}

Hence the auxiliary function is initially positive.

Next, consider the asymptotic limit $x\rightarrow\infty$. Since the highest-order term dominates,
\begin{equation}
g(x) = 1 + (q-1)A_qx^{q-1} - (2q-1)B_qx^{2q-1},
\end{equation}
and because
\begin{equation}
2q-1>q-1,
\end{equation}
the negative contribution dominates for sufficiently large $x$. Therefore,
\begin{equation}
\boxed{\lim_{x\to\infty}g(x) = -\infty.}
\end{equation}

Thus the auxiliary function starts from
\begin{equation}
g(0^+)=1>0,
\end{equation}
attains a unique global maximum at $x=x_c$, and subsequently decreases monotonically toward $-\infty$.

Since $g(x)$ is continuous on $(0,\infty)$, the intermediate value theorem guarantees the existence of at least one positive root.

Furthermore, because the function is strictly increasing on $(0,x_c)$ and strictly decreasing on $(x_c,\infty)$, no more than one intersection with the horizontal axis is possible. Consequently, the positive root is unique.

Hence
\begin{equation}
\boxed{g(x)=0 \text{ possesses exactly one positive solution.}}
\end{equation}

Since
\begin{equation}
g(x)=xf'(x),
\end{equation}
the stationary function $f(x)$ possesses exactly one critical point.

Therefore the graph of $f(x)$ is unimodal. Because
\begin{equation}
\lim_{x\to0^+}f(x) = -\infty, \qquad \lim_{x\to\infty}f(x) = -\infty,
\end{equation}
a unimodal stationary function can intersect the horizontal axis at most twice.

Consequently,
\begin{equation}
\boxed{f(x)=0 \text{ possesses at most two positive roots.}}
\end{equation}

This establishes the fundamental geometrical restriction underlying the two-valued structure of the complexity-maximizing distributions.

\subsection{Exclusion of the zero-root and one-root cases}

The previous subsection established that the stationary equation
\begin{equation}
f(x)=0
\end{equation}
possesses at most two positive roots. We now show that the possibilities of zero roots and one root are both impossible.

\paragraph{Exclusion of the zero-root case.}

Suppose that the stationary equation possesses no positive roots. Then no interior stationary point exists, and therefore the maximum of complexity must necessarily occur on the boundary of the probability simplex.

However, the two extremal configurations are the uniform distribution
\begin{equation}
W_{\mathrm{u}} = \left( \frac{1}{N},\ldots,\frac{1}{N} \right),
\end{equation}
and the Dirac distribution
\begin{equation}
W_{\mathrm{d}} = (1,0,\ldots,0).
\end{equation}

Both configurations satisfy
\begin{equation}
\Phi(W_{\mathrm{u}}) = \Phi(W_{\mathrm{d}}) = 0.
\end{equation}

Consequently,
\begin{equation}
\Phi_{\max}=0.
\end{equation}

On the other hand, there exist nontrivial probability distributions for which
\begin{equation}
\Phi>0.
\end{equation}

Hence
\begin{equation}
\Phi_{\max}>0,
\end{equation}
which contradicts the assumption that the maximizer lies on the boundary.

Therefore,
\begin{equation}
\boxed{N_r\neq0,}
\end{equation}
where $N_r$ denotes the number of positive roots of the stationary equation.

\paragraph{Exclusion of the one-root case.}

Suppose next that the stationary equation possesses exactly one positive root, denoted by $x_0$.

Since every probability component satisfies
\begin{equation}
f(w_i)=0,
\end{equation}
all components must coincide:
\begin{equation}
w_1=w_2=\cdots=w_N=x_0.
\end{equation}

Normalization implies
\begin{equation}
Nx_0=1,
\end{equation}
and therefore
\begin{equation}
x_0=\frac{1}{N}.
\end{equation}

Hence the one-root case corresponds uniquely to the uniform distribution.

But the uniform distribution satisfies
\begin{equation}
\Phi=0,
\end{equation}
which again contradicts the existence of probability distributions with positive complexity.

Consequently,
\begin{equation}
\boxed{N_r\neq1.}
\end{equation}

\paragraph{The two-root theorem.}

Since the previous subsection established that
\begin{equation}
N_r\le2,
\end{equation}
while the preceding arguments imply
\begin{equation}
N_r\neq0, \qquad N_r\neq1,
\end{equation}
the only remaining possibility is
\begin{equation}
\boxed{N_r=2.}
\end{equation}

Therefore the stationary equation possesses exactly two positive solutions,
\begin{equation}
x_h<x_c,
\end{equation}
and every component of a complexity-maximizing distribution must assume one of these two values.

Consequently, the complexity maximizer necessarily possesses the form
\begin{equation}
W = ( \underbrace{x_c,\ldots,x_c}_{m}, \underbrace{x_h,\ldots,x_h}_{N-m} ),
\end{equation}
thereby establishing the origin of the core--halo structure.

\subsection{Continuous analogue of the two-root theorem}

The preceding arguments extend naturally to the continuous case.

Let
\begin{equation}
\rho(x)\ge0,
\end{equation}
be a normalized probability density satisfying
\begin{equation}
\int_{\Omega}\rho(x)\,dx=1.
\end{equation}

The stationary equation derived in Appendix A assumes the form
\begin{equation}
\ln\rho(x) + A_q\rho(x)^{q-1} - B_q\rho(x)^{2q-1} + C = 0.
\end{equation}

As in the discrete case, define the stationary function
\begin{equation}
f(x) = \ln x + A_qx^{q-1} - B_qx^{2q-1} + C,
\end{equation}
whose positive roots determine the possible values assumed by the density.

From the previous subsections, we know that the equation
\begin{equation}
f(x)=0
\end{equation}
possesses at most two positive roots.

We now exclude the possibilities of zero and one roots.

\paragraph{Exclusion of the zero-root case.}

Suppose that no positive roots exist. Then no interior stationary density exists, and therefore the maximum of complexity must occur on the boundary of the space of probability densities.

Two extremal configurations are immediately identified:

\begin{enumerate}
\item the uniform density
\begin{equation}
\rho_u(x)=\frac{1}{|\Omega|},
\end{equation}

\item the singular Dirac distribution
\begin{equation}
\rho_d(x)=\delta(x-x_0).
\end{equation}
\end{enumerate}

Both configurations satisfy
\begin{equation}
\Phi[\rho_u] = \Phi[\rho_d] = 0.
\end{equation}

Hence one would obtain
\begin{equation}
\Phi_{\max}=0,
\end{equation}
which contradicts the existence of nontrivial probability densities for which
\begin{equation}
\Phi>0.
\end{equation}

Therefore,
\begin{equation}
N_r\neq0.
\end{equation}

\paragraph{Exclusion of the one-root case.}

Suppose now that the stationary equation possesses exactly one positive root, denoted by $x_0$.

Since every point of the support satisfies the same stationary equation, one must have
\begin{equation}
\rho(x)=x_0,
\end{equation}
for all $x\in\Omega$.

Normalization yields
\begin{equation}
x_0|\Omega|=1,
\end{equation}
and therefore
\begin{equation}
x_0=\frac{1}{|\Omega|}.
\end{equation}

Hence the one-root case corresponds uniquely to the uniform density.

Since the uniform density satisfies
\begin{equation}
\Phi=0,
\end{equation}
while nonuniform densities exist for which
\begin{equation}
\Phi>0,
\end{equation}
the one-root case is impossible. Therefore,
\begin{equation}
N_r\neq1.
\end{equation}

Combining these results with the fact that
\begin{equation}
N_r\le2,
\end{equation}
we conclude that
\begin{equation}
\boxed{N_r=2.}
\end{equation}

Hence the continuous stationary equation possesses exactly two positive branches,
\begin{equation}
\rho_h<\rho_c,
\end{equation}
which constitute the continuous core--halo structure.

Thus the two-valued theorem established for discrete probability vectors remains valid in the continuum limit.

\subsection{Reduction of the variational problem to multiplicity analysis}

The preceding subsections established that the stationary equation
\begin{equation}
f(x)=0
\end{equation}
possesses exactly two positive roots,
\begin{equation}
x_h<x_c.
\end{equation}

Consequently, every probability component must assume one of these two values.

\paragraph{Discrete case.}

Let $m$ denote the number of components occupying the upper branch $x_c$. Then the probability vector necessarily possesses the form
\begin{equation}
W = ( \underbrace{x_c,\ldots,x_c}_{m}, \underbrace{x_h,\ldots,x_h}_{N-m} ).
\end{equation}

Normalization imposes
\begin{equation}
mx_c+(N-m)x_h=1.
\end{equation}

Therefore, once the two stationary branches have been determined, the original optimization problem over the entire simplex reduces to the determination of the integer parameter
\begin{equation}
m\in\{1,2,\ldots,N-1\}.
\end{equation}

Hence the variational problem becomes a discrete combinatorial problem.

\paragraph{Continuous case.}

For a continuous density, let $\Omega_c \subset\Omega$ denote the subset occupied by the core branch $\rho_c$, and let $\Omega_h=\Omega\setminus\Omega_c$ denote the complementary halo region.

Define
\begin{equation}
M=|\Omega_c|, \qquad |\Omega_h| = |\Omega|-M.
\end{equation}

The stationary density therefore assumes the form
\begin{equation}
\rho(x) = \begin{cases} \rho_c, & x\in\Omega_c, \\[1ex] \rho_h, & x\in\Omega_h. \end{cases}
\end{equation}

Normalization gives
\begin{equation}
M\rho_c+(|\Omega|-M)\rho_h=1.
\end{equation}

Hence the infinite-dimensional variational problem over all admissible densities is reduced to the determination of the measure parameter $M$.

\paragraph{Unified formulation.}

Thus both the discrete and continuous optimization problems are transformed into multiplicity problems.

The stationary amplitudes $x_h, x_c$ (or equivalently $\rho_h, \rho_c$) are determined by the stationary equation, while the remaining degree of freedom is encoded entirely in the multiplicity parameter $m$ for discrete systems, or the measure parameter $M$ for continuous systems.

Therefore, the determination of the complexity maximizer reduces to the study of the function
\begin{equation}
\Phi=\Phi(m)
\end{equation}
in the discrete case and
\begin{equation}
\Phi=\Phi(M)
\end{equation}
in the continuous case.

The analysis of these functions constitutes the mathematical foundation of the $m$-core theorem proved in Section~7.

\section{Derivation of the multiplicity optimization}

\subsection{Reduction of the constrained optimization problem}

The generalized LMC complexity introduced in the main text is

\begin{equation}
\Phi
=
H_{1}
+
2\ln P_{q}
-
\ln P_{2q},
\label{eq:C1_phi_original}
\end{equation}

where

\begin{align}
H_{1}
&=
-\sum_{i=1}^{N}w_i\ln w_i,
\\
P_q
&=
\sum_{i=1}^{N}w_i^q,
\\
P_{2q}
&=
\sum_{i=1}^{N}w_i^{2q},
\end{align}

subject to the normalization constraint

\begin{equation}
\sum_{i=1}^{N}w_i=1.
\label{eq:C1_normalization}
\end{equation}

In Section~6 it was proved that every stationary point of
\(\Phi\) possesses exactly two distinct probability values.
Consequently every extremum belongs to the two-level family

\begin{equation}
W_m
=
(
\underbrace{w_c,\ldots,w_c}_{m},
\underbrace{w_h,\ldots,w_h}_{N-m}
),
\label{eq:C1_twolevel}
\end{equation}

where \(m\in\{1,\ldots,N-1\}\).

The normalization condition immediately becomes

\begin{equation}
mw_c+(N-m)w_h=1.
\label{eq:C1_constraint}
\end{equation}

Substituting (\ref{eq:C1_twolevel}) into the complexity functional gives

\begin{align}
\Phi
&=
-mw_c\ln w_c
-(N-m)w_h\ln w_h
\nonumber\\
&
\qquad
+
2\ln
\!\left(
mw_c^{\,q}
+
(N-m)w_h^{\,q}
\right)
\nonumber\\
&
\qquad
-
\ln
\!\left(
mw_c^{\,2q}
+
(N-m)w_h^{\,2q}
\right).
\label{eq:C1_phi_twolevel}
\end{align}

To enforce the normalization condition we introduce the Lagrangian

\begin{equation}
\mathcal L
=
\Phi
+
\lambda
\left[
1-mw_c-(N-m)w_h
\right].
\label{eq:C1_lagrangian}
\end{equation}

Equation (\ref{eq:C1_constraint}) allows one variable to be eliminated.
Solving for the halo amplitude gives

\begin{equation}
w_h
=
\frac{1-mw_c}{N-m}.
\label{eq:C1_wh}
\end{equation}

Substituting (\ref{eq:C1_wh}) into
(\ref{eq:C1_lagrangian}) causes the constraint term to vanish identically,

\[
1-mw_c-(N-m)
\left(
\frac{1-mw_c}{N-m}
\right)
=0,
\]

so that

\begin{equation}
\boxed{
\mathcal L
=
\Phi(w_c,m).
}
\label{eq:C1_unconstrained}
\end{equation}

Hence the constrained optimization problem has been transformed into an
ordinary optimization problem involving only the variables
\((w_c,m)\).

For every fixed multiplicity \(m\), the optimal core amplitude is obtained
from the first-order stationarity condition

\begin{equation}
\frac{\partial\Phi}{\partial w_c}=0,
\label{eq:C1_stationary}
\end{equation}

whose solution defines a unique function

\begin{equation}
w_c=w_c^{*}(m).
\end{equation}

Substituting this stationary solution back into the objective produces the
optimized one-parameter family

\begin{equation}
\boxed{
\Phi^{*}(m)
=
\Phi\!\left(w_c^{*}(m),m\right),
}
\label{eq:C1_phi_star}
\end{equation}

which is the object studied throughout the multiplicity optimization
analysis.

\subsubsection*{Continuous formulation}

The continuous generalized LMC complexity is

\begin{equation}
\Phi
=
H_{1}
+
2\ln P_q
-
\ln P_{2q},
\end{equation}

where

\begin{align}
H_{1}
&=
-\int_{\Omega}\rho(x)\ln\rho(x)\,dx,
\\
P_q
&=
\int_{\Omega}\rho(x)^q\,dx,
\\
P_{2q}
&=
\int_{\Omega}\rho(x)^{2q}\,dx,
\end{align}

subject to the normalization condition

\begin{equation}
\int_{\Omega}\rho(x)\,dx=1.
\end{equation}

By the continuous two-level theorem proved in Section~6, every stationary
density possesses exactly two distinct probability values. Consequently,
every extremum belongs to the family

\begin{equation}
\rho_M(x)
=
\begin{cases}
\rho_c, & x\in\Omega_c,\\
\rho_h, & x\in\Omega_h,
\end{cases}
\end{equation}

where

\[
\mu(\Omega_c)=M,
\qquad
\mu(\Omega_h)=L-M,
\]

and \(L=\mu(\Omega)\) denotes the total measure of the support.

The normalization condition therefore becomes

\begin{equation}
M\rho_c+(L-M)\rho_h=1.
\label{eq:C1_constraint_cont}
\end{equation}

Substituting the two-level density into the functional gives

\begin{align}
\Phi
&=
-M\rho_c\ln\rho_c
-(L-M)\rho_h\ln\rho_h
\nonumber\\
&
\qquad
+
2\ln
\left(
M\rho_c^{\,q}
+
(L-M)\rho_h^{\,q}
\right)
\nonumber\\
&
\qquad
-
\ln
\left(
M\rho_c^{\,2q}
+
(L-M)\rho_h^{\,2q}
\right).
\end{align}

Introducing the constrained Lagrangian,

\begin{equation}
\mathcal L
=
\Phi
+
\lambda
\left[
1
-
M\rho_c
-
(L-M)\rho_h
\right],
\end{equation}

and eliminating the halo amplitude through

\begin{equation}
\rho_h
=
\frac{1-M\rho_c}{L-M},
\end{equation}

the constraint term again vanishes identically, yielding

\begin{equation}
\boxed{
\mathcal L
=
\Phi(\rho_c,M).
}
\end{equation}

Hence the continuous constrained optimization problem also reduces to an
ordinary optimization problem involving only the variables
\((\rho_c,M)\).

For every fixed measure \(M\), the optimal core density is determined by

\begin{equation}
\frac{\partial\Phi}{\partial\rho_c}=0,
\end{equation}

whose solution defines

\begin{equation}
\rho_c=\rho_c^{*}(M).
\end{equation}

Substitution into the objective produces the optimized one-parameter family

\begin{equation}
\boxed{
\Phi^{*}(M)
=
\Phi\!\left(\rho_c^{*}(M),M\right).
}
\end{equation}

\subsection{Derivation of the stationary equation}

Starting from the unconstrained functional obtained in the previous subsection,
\begin{equation}
\Phi=\Phi(w_c,m),
\end{equation}
the optimal core amplitude is determined by the first-order stationarity
condition
\begin{equation}
\frac{\partial\Phi}{\partial w_c}=0.
\label{eq:C2_foc_start}
\end{equation}

Since the halo amplitude is constrained by
\begin{equation}
w_h=\frac{1-mw_c}{N-m},
\end{equation}
its derivative with respect to the core amplitude is
\begin{equation}
\frac{dw_h}{dw_c}
=
-\frac{m}{N-m}.
\label{eq:C2_dwh}
\end{equation}

The Shannon entropy contribution becomes
\begin{align}
H_1
&=
-mw_c\ln w_c
-(N-m)w_h\ln w_h.
\end{align}

Differentiating with respect to \(w_c\) gives
\begin{align}
\frac{\partial H_1}{\partial w_c}
&=
-m(\ln w_c+1)
-(N-m)
(\ln w_h+1)
\frac{dw_h}{dw_c}
\nonumber\\
&=
m
\ln\!\left(\frac{w_h}{w_c}\right).
\label{eq:C2_dH}
\end{align}

Next consider the escort moments
\begin{align}
P_q
&=
mw_c^q
+
(N-m)w_h^q,
\\
P_{2q}
&=
mw_c^{2q}
+
(N-m)w_h^{2q}.
\end{align}

Using (\ref{eq:C2_dwh}),
\begin{align}
\frac{\partial P_q}{\partial w_c}
&=
mqw_c^{q-1}
+
(N-m)qw_h^{q-1}
\frac{dw_h}{dw_c}
\nonumber\\
&=
mq
\left(
w_c^{q-1}
-
w_h^{q-1}
\right),
\label{eq:C2_dPq}
\end{align}
and similarly
\begin{align}
\frac{\partial P_{2q}}{\partial w_c}
&=
2qm
\left(
w_c^{2q-1}
-
w_h^{2q-1}
\right).
\label{eq:C2_dP2q}
\end{align}

Applying the chain rule,

\begin{align}
\frac{\partial}{\partial w_c}
\left(
2\ln P_q
\right)
&=
\frac{2}{P_q}
\frac{\partial P_q}{\partial w_c}
\nonumber\\
&=
\frac{2qm
\left(
w_c^{q-1}
-
w_h^{q-1}
\right)}
{P_q},
\end{align}

and

\begin{align}
\frac{\partial}{\partial w_c}
\left(
-\ln P_{2q}
\right)
&=
-
\frac{1}{P_{2q}}
\frac{\partial P_{2q}}{\partial w_c}
\nonumber\\
&=
-
\frac{2qm
\left(
w_c^{2q-1}
-
w_h^{2q-1}
\right)}
{P_{2q}}.
\end{align}

Collecting all contributions,

\begin{align}
\frac{\partial\Phi}{\partial w_c}
&=
m
\ln\!\left(\frac{w_h}{w_c}\right)
+
\frac{2qm
\left(
w_c^{q-1}
-
w_h^{q-1}
\right)}
{P_q}
\nonumber\\
&
\qquad
-
\frac{2qm
\left(
w_c^{2q-1}
-
w_h^{2q-1}
\right)}
{P_{2q}}.
\end{align}

Since \(m>0\), the stationarity condition
(\ref{eq:C2_foc_start}) is equivalent to

\begin{equation}
\boxed{
\ln\!\left(\frac{w_c}{w_h}\right)
=
-
\frac{2q
\left(
w_c^{q-1}
-
w_h^{q-1}
\right)}
{P_q}
+
\frac{2q
\left(
w_c^{2q-1}
-
w_h^{2q-1}
\right)}
{P_{2q}}.
}
\label{eq:C2_foc}
\end{equation}

Equation (\ref{eq:C2_foc}) is precisely the stationary equation used
throughout the multiplicity analysis. Solving this equation determines the
optimal core amplitude \(w_c^*(m)\), from which the optimized complexity
\(\Phi^*(m)\) follows immediately.

\subsubsection*{Continuous derivation}

Starting from the unconstrained functional

\begin{equation}
\Phi=\Phi(\rho_c,M),
\end{equation}

the optimal core density is determined by the stationarity condition

\begin{equation}
\frac{\partial\Phi}{\partial\rho_c}=0.
\label{eq:C2_foc_start_cont}
\end{equation}

Using the normalization condition

\begin{equation}
M\rho_c+(L-M)\rho_h=1,
\end{equation}

the halo density satisfies

\begin{equation}
\rho_h
=
\frac{1-M\rho_c}{L-M},
\end{equation}

and therefore

\begin{equation}
\frac{d\rho_h}{d\rho_c}
=
-\frac{M}{L-M}.
\label{eq:C2_drhoh}
\end{equation}

The Shannon contribution is

\begin{equation}
H_1
=
-M\rho_c\ln\rho_c
-
(L-M)\rho_h\ln\rho_h.
\end{equation}

Differentiating with respect to \(\rho_c\),

\begin{align}
\frac{\partial H_1}{\partial\rho_c}
&=
-M(\ln\rho_c+1)
-(L-M)(\ln\rho_h+1)
\frac{d\rho_h}{d\rho_c}
\nonumber\\
&=
M
\ln\!\left(\frac{\rho_h}{\rho_c}\right).
\label{eq:C2_dH_cont}
\end{align}

The escort moments are

\begin{align}
P_q
&=
M\rho_c^q
+
(L-M)\rho_h^q,
\\
P_{2q}
&=
M\rho_c^{2q}
+
(L-M)\rho_h^{2q}.
\end{align}

Their derivatives become

\begin{align}
\frac{\partial P_q}{\partial\rho_c}
&=
Mq
\left(
\rho_c^{q-1}
-
\rho_h^{q-1}
\right),
\label{eq:C2_dPq_cont}
\\
\frac{\partial P_{2q}}{\partial\rho_c}
&=
2Mq
\left(
\rho_c^{2q-1}
-
\rho_h^{2q-1}
\right).
\label{eq:C2_dP2q_cont}
\end{align}

Using the chain rule,

\begin{align}
\frac{\partial}{\partial\rho_c}
\left(
2\ln P_q
\right)
&=
\frac{
2Mq
\left(
\rho_c^{q-1}
-
\rho_h^{q-1}
\right)
}{P_q},
\end{align}

and

\begin{align}
\frac{\partial}{\partial\rho_c}
\left(
-\ln P_{2q}
\right)
&=
-
\frac{
2Mq
\left(
\rho_c^{2q-1}
-
\rho_h^{2q-1}
\right)
}{P_{2q}}.
\end{align}

Collecting every contribution,

\begin{align}
\frac{\partial\Phi}{\partial\rho_c}
&=
M
\ln\!\left(\frac{\rho_h}{\rho_c}\right)
+
\frac{
2Mq
\left(
\rho_c^{q-1}
-
\rho_h^{q-1}
\right)
}{P_q}
\nonumber\\
&\qquad
-
\frac{
2Mq
\left(
\rho_c^{2q-1}
-
\rho_h^{2q-1}
\right)
}{P_{2q}}.
\end{align}

Since \(M>0\), the stationarity condition
(\ref{eq:C2_foc_start_cont}) is equivalent to

\begin{equation}
\boxed{
\ln\!\left(\frac{\rho_c}{\rho_h}\right)
=
-
\frac{
2q
\left(
\rho_c^{q-1}
-
\rho_h^{q-1}
\right)
}{P_q}
+
\frac{
2q
\left(
\rho_c^{2q-1}
-
\rho_h^{2q-1}
\right)
}{P_{2q}}.
}
\label{eq:C2_foc_cont}
\end{equation}

Equation (\ref{eq:C2_foc_cont}) is the continuous counterpart of
(\ref{eq:C2_foc}) and determines the optimal core density
\(\rho_c^*(M)\). Substituting this solution into the functional yields the
optimized one-parameter family \(\Phi^*(M)\).

\subsection{Application of the Envelope Theorem}

The stationary equation (\ref{eq:C2_foc}) determines the optimal core
probability as an implicit function of the multiplicity parameter,
\begin{equation}
w_c=w_c^{*}(m).
\end{equation}

Consequently, the optimized generalized complexity is
\begin{equation}
\Phi^{*}(m)
=
\Phi\!\left(w_c^{*}(m),m\right),
\end{equation}
and its total derivative with respect to the multiplicity is obtained by the
chain rule,
\begin{equation}
\frac{d\Phi^{*}}{dm}
=
\frac{\partial\Phi}{\partial m}
+
\frac{\partial\Phi}{\partial w_c}
\frac{dw_c^{*}}{dm}.
\label{eq:C3_chain}
\end{equation}

However, the optimized branch \(w_c^{*}(m)\) satisfies the stationary
condition
\begin{equation}
\frac{\partial\Phi}{\partial w_c}=0.
\label{eq:C3_stationary}
\end{equation}

Substituting (\ref{eq:C3_stationary}) into
(\ref{eq:C3_chain}) immediately gives

\begin{equation}
\boxed{
\frac{d\Phi^{*}}{dm}
=
\frac{\partial\Phi}{\partial m}.
}
\label{eq:C3_envelope}
\end{equation}

Equation (\ref{eq:C3_envelope}) is the classical envelope theorem specialized
to the present optimization problem. It shows that the implicit variation of
the stationary amplitude contributes nothing to the derivative of the
optimized complexity.

Therefore, the entire dependence of the optimized complexity upon the
multiplicity parameter is obtained simply by differentiating the original
functional while treating the stationary amplitudes as constants.

This observation transforms the optimization problem from the differentiation
of an implicitly defined function into the evaluation of an ordinary partial
derivative. The explicit computation of this derivative is carried out in the
following subsection.

\subsubsection*{Continuous analogue}

The stationary equation (\ref{eq:C2_foc_cont}) determines the optimal core
density as an implicit function of the core measure,

\begin{equation}
\rho_c=\rho_c^{*}(M).
\end{equation}

Consequently, the optimized generalized complexity is

\begin{equation}
\Phi^{*}(M)
=
\Phi\!\left(\rho_c^{*}(M),M\right),
\end{equation}

whose total derivative with respect to the core measure is

\begin{equation}
\frac{d\Phi^{*}}{dM}
=
\frac{\partial\Phi}{\partial M}
+
\frac{\partial\Phi}{\partial\rho_c}
\frac{d\rho_c^{*}}{dM}.
\label{eq:C3_chain_cont}
\end{equation}

The optimized branch satisfies the stationary condition

\begin{equation}
\frac{\partial\Phi}{\partial\rho_c}=0.
\label{eq:C3_stationary_cont}
\end{equation}

Hence,

\begin{equation}
\boxed{
\frac{d\Phi^{*}}{dM}
=
\frac{\partial\Phi}{\partial M}.
}
\label{eq:C3_envelope_cont}
\end{equation}

Thus, exactly as in the discrete problem, the implicit dependence of the
stationary density upon the core measure contributes nothing to the variation
of the optimized complexity.

The optimization therefore reduces to evaluating the explicit partial
derivative of the original functional while treating the stationary densities
as constants. This derivative is computed in the next subsection.

\subsection{Explicit evaluation of the envelope derivative}

By the envelope theorem,
\begin{equation}
\frac{d\Phi^{*}}{dm}
=
\frac{\partial\Phi}{\partial m},
\end{equation}
where the partial derivative is evaluated while holding the stationary core
probability \(w_c\) fixed.

Since
\begin{equation}
w_h=\frac{1-mw_c}{N-m},
\end{equation}
its partial derivative with respect to the multiplicity is

\begin{equation}
\frac{\partial w_h}{\partial m}
=
\frac{w_h-w_c}{N-m}.
\label{eq:C4_dwhdm}
\end{equation}

The Shannon entropy contribution is

\begin{equation}
H_1
=
-mw_c\ln w_c
-
(N-m)w_h\ln w_h.
\end{equation}

Using (\ref{eq:C4_dwhdm}), one finds

\begin{align}
\frac{\partial H_1}{\partial m}
&=
(w_c-w_h)
-
w_c
\ln\!\left(\frac{w_c}{w_h}\right).
\label{eq:C4_entropy}
\end{align}

Next consider

\begin{equation}
P_q
=
mw_c^q
+
(N-m)w_h^q.
\end{equation}

Differentiating with respect to \(m\),

\begin{align}
\frac{\partial P_q}{\partial m}
&=
w_c^q
-
w_h^q
+
(N-m)
q
w_h^{q-1}
\frac{\partial w_h}{\partial m}
\nonumber\\
&=
q\,w_h^{q-1}(w_c-w_h)
-
(w_c^q-w_h^q).
\label{eq:C4_dPq}
\end{align}

Similarly,

\begin{align}
\frac{\partial P_{2q}}{\partial m}
&=
w_c^{2q}
-
w_h^{2q}
+
(N-m)
(2q)
w_h^{2q-1}
\frac{\partial w_h}{\partial m}
\nonumber\\
&=
(w_c^{2q}-w_h^{2q})
-
2q\,w_h^{2q-1}(w_c-w_h).
\label{eq:C4_dP2q}
\end{align}

Applying the chain rule,

\begin{align}
\frac{\partial}{\partial m}
\left(
2\ln P_q
\right)
&=
\frac{2}{P_q}
\frac{\partial P_q}{\partial m},
\\
\frac{\partial}{\partial m}
\left(
-\ln P_{2q}
\right)
&=
-
\frac{1}{P_{2q}}
\frac{\partial P_{2q}}{\partial m}.
\end{align}

Collecting all contributions gives the exact envelope derivative

\begin{align}
\frac{d\Phi^{*}}{dm}
&=
(w_c-w_h)
-
w_c
\ln\!\left(\frac{w_c}{w_h}\right)
\nonumber\\
&\quad
+
\frac{
2\left[
q\,w_h^{q-1}(w_c-w_h)
-
(w_c^q-w_h^q)
\right]
}{P_q}
\nonumber\\
&\quad
+
\frac{
(w_c^{2q}-w_h^{2q})
-
2q\,w_h^{2q-1}(w_c-w_h)
}{P_{2q}},
\label{eq:C4_raw_envelope}
\end{align}

where

\begin{equation}
P_q
=
mw_c^q+(N-m)w_h^q,
\qquad
P_{2q}
=
mw_c^{2q}+(N-m)w_h^{2q}.
\end{equation}

Equation (\ref{eq:C4_raw_envelope}) is the envelope derivative before eliminating the logarithmic contribution. In the following subsection the logarithm is removed algebraically using the stationary equation, yielding the purely algebraic expression used in the main proof.

\subsubsection*{Continuous analogue}

By the envelope theorem,

\begin{equation}
\frac{d\Phi^{*}}{dM}
=
\frac{\partial\Phi}{\partial M},
\end{equation}

where the partial derivative is evaluated while holding the stationary core
density \(\rho_c\) fixed.

Since

\begin{equation}
\rho_h
=
\frac{1-M\rho_c}{L-M},
\end{equation}

its partial derivative with respect to the core measure is

\begin{equation}
\frac{\partial\rho_h}{\partial M}
=
\frac{\rho_h-\rho_c}{L-M}.
\label{eq:C4_drhohdM}
\end{equation}

The Shannon entropy contribution is

\begin{equation}
H_1
=
-M\rho_c\ln\rho_c
-
(L-M)\rho_h\ln\rho_h.
\end{equation}

Using (\ref{eq:C4_drhohdM}), one finds

\begin{align}
\frac{\partial H_1}{\partial M}
&=
(\rho_c-\rho_h)
-
\rho_c
\ln\!\left(\frac{\rho_c}{\rho_h}\right).
\label{eq:C4_entropy_cont}
\end{align}

Next consider

\begin{equation}
P_q
=
M\rho_c^q
+
(L-M)\rho_h^q.
\end{equation}

Differentiating with respect to \(M\),

\begin{align}
\frac{\partial P_q}{\partial M}
&=
q\rho_h^{q-1}(\rho_c-\rho_h)
-
(\rho_c^q-\rho_h^q).
\label{eq:C4_dPq_cont}
\end{align}

Similarly,

\begin{align}
\frac{\partial P_{2q}}{\partial M}
&=
2q\rho_h^{2q-1}(\rho_c-\rho_h)
-
(\rho_c^{2q}-\rho_h^{2q}).
\label{eq:C4_dP2q_cont}
\end{align}

Applying the chain rule,

\begin{align}
\frac{\partial}{\partial M}
\left(
2\ln P_q
\right)
&=
\frac{2}{P_q}
\frac{\partial P_q}{\partial M},
\\
\frac{\partial}{\partial M}
\left(
-\ln P_{2q}
\right)
&=
-
\frac{1}{P_{2q}}
\frac{\partial P_{2q}}{\partial M}.
\end{align}

Collecting all contributions gives the exact envelope derivative

\begin{align}
\frac{d\Phi^{*}}{dM}
&=
(\rho_c-\rho_h)
-
\rho_c
\ln\!\left(\frac{\rho_c}{\rho_h}\right)
\nonumber\\
&\quad
+
\frac{
2\left[
q\rho_h^{q-1}(\rho_c-\rho_h)
-
(\rho_c^q-\rho_h^q)
\right]
}{P_q}
\nonumber\\
&\quad
+
\frac{
(\rho_c^{2q}-\rho_h^{2q})
-
2q\rho_h^{2q-1}(\rho_c-\rho_h)
}{P_{2q}},
\label{eq:C4_raw_envelope_cont}
\end{align}

where

\begin{equation}
P_q
=
M\rho_c^q
+
(L-M)\rho_h^q,
\qquad
P_{2q}
=
M\rho_c^{2q}
+
(L-M)\rho_h^{2q}.
\end{equation}

Equation (\ref{eq:C4_raw_envelope_cont}) is the continuous counterpart of
(\ref{eq:C4_raw_envelope}). In the following subsection, the logarithmic term
is eliminated using the stationary equation, reducing the envelope derivative
to the same universal algebraic structure obtained in the discrete case.

\subsection{Elimination of the logarithmic contribution}

The envelope derivative obtained in the previous subsection still contains the
transcendental quantity
\(
\ln\!\left(\frac{w_c}{w_h}\right).
\)
This logarithm may be removed exactly by exploiting the first-order optimality
condition.

The stationary equation derived from
\(
\partial\Phi/\partial w_c=0
\)
can be written as

\begin{equation}
\ln\!\left(\frac{w_c}{w_h}\right)
=
-\frac{2q\left(w_c^{q-1}-w_h^{q-1}\right)}{P_q}
+
\frac{2q\left(w_c^{2q-1}-w_h^{2q-1}\right)}{P_{2q}}.
\label{eq:C5_foc}
\end{equation}

Substituting (\ref{eq:C5_foc}) into
(\ref{eq:C4_raw_envelope})
produces

\begin{align}
\frac{d\Phi^{*}}{dm}
&=
(w_c-w_h)
\nonumber\\
&\quad
+
\frac{
2\!\left[
q\,w_c\!\left(w_c^{q-1}-w_h^{q-1}\right)
+
q\,w_h^{q-1}(w_c-w_h)
-
(w_c^q-w_h^q)
\right]
}{P_q}
\nonumber\\
&\quad
+
\frac{
(w_c^{2q}-w_h^{2q})
-
2q\,w_c\!\left(w_c^{2q-1}-w_h^{2q-1}\right)
-
2q\,w_h^{2q-1}(w_c-w_h)
}{P_{2q}}.
\end{align}

The numerators simplify identically,

\begin{align}
&q\,w_c\!\left(w_c^{q-1}-w_h^{q-1}\right)
+
q\,w_h^{q-1}(w_c-w_h)
-
(w_c^q-w_h^q)
\nonumber\\
&\qquad
=
(q-1)\left(w_c^q-w_h^q\right),
\end{align}

and

\begin{align}
&(w_c^{2q}-w_h^{2q})
-
2q\,w_c\!\left(w_c^{2q-1}-w_h^{2q-1}\right)
-
2q\,w_h^{2q-1}(w_c-w_h)
\nonumber\\
&\qquad
=
-(2q-1)\left(w_c^{2q}-w_h^{2q}\right).
\end{align}

Consequently, the envelope derivative reduces to the purely algebraic form

\begin{equation}
\boxed{
\frac{d\Phi^{*}}{dm}
=
(w_c-w_h)
+
\frac{2(q-1)\left(w_c^q-w_h^q\right)}{P_q}
-
\frac{(2q-1)\left(w_c^{2q}-w_h^{2q}\right)}{P_{2q}}.
}
\label{eq:C5_logfree}
\end{equation}

Unlike (\ref{eq:C4_raw_envelope}), Equation
(\ref{eq:C5_logfree}) contains no logarithmic terms.
The reduction is exact and follows solely from the first-order optimality
condition. This identity was further verified independently by symbolic
computation using the \texttt{SymPy} computer algebra system, which confirms
that the difference between the raw envelope derivative and the
logarithm-free expression simplifies identically to zero after substitution of
(\ref{eq:C5_foc}).

\subsubsection*{Continuous analogue}

The envelope derivative obtained in the previous subsection still contains the
transcendental quantity

\[
\ln\!\left(\frac{\rho_c}{\rho_h}\right).
\]

This logarithm may be eliminated exactly by using the first-order optimality
condition.

The stationary equation derived from

\[
\frac{\partial\Phi}{\partial\rho_c}=0
\]

is

\begin{equation}
\ln\!\left(\frac{\rho_c}{\rho_h}\right)
=
-\frac{2q\left(\rho_c^{q-1}-\rho_h^{q-1}\right)}{P_q}
+
\frac{2q\left(\rho_c^{2q-1}-\rho_h^{2q-1}\right)}{P_{2q}}.
\label{eq:C5_foc_cont}
\end{equation}

Substituting (\ref{eq:C5_foc_cont}) into
(\ref{eq:C4_raw_envelope_cont}) gives

\begin{align}
\frac{d\Phi^{*}}{dM}
&=
(\rho_c-\rho_h)
\nonumber\\
&
\quad
+
\frac{
2q\rho_c
\left(
\rho_c^{q-1}-\rho_h^{q-1}
\right)
+
2
\left[
q\rho_h^{q-1}(\rho_c-\rho_h)
-
(\rho_c^q-\rho_h^q)
\right]
}{P_q}
\nonumber\\
&
\quad
-
\frac{
2q\rho_c
\left(
\rho_c^{2q-1}-\rho_h^{2q-1}
\right)
-
(\rho_c^{2q}-\rho_h^{2q})
+
2q\rho_h^{2q-1}(\rho_c-\rho_h)
}{P_{2q}}.
\end{align}

The numerators simplify identically,

\begin{align}
&
2q\rho_c
\left(
\rho_c^{q-1}-\rho_h^{q-1}
\right)
+
2
\left[
q\rho_h^{q-1}(\rho_c-\rho_h)
-
(\rho_c^q-\rho_h^q)
\right]
\nonumber\\
&\qquad
=
2(q-1)
\left(
\rho_c^q-\rho_h^q
\right),
\end{align}

and

\begin{align}
&
2q\rho_c
\left(
\rho_c^{2q-1}-\rho_h^{2q-1}
\right)
-
(\rho_c^{2q}-\rho_h^{2q})
+
2q\rho_h^{2q-1}
(\rho_c-\rho_h)
\nonumber\\
&\qquad
=
(2q-1)
\left(
\rho_c^{2q}-\rho_h^{2q}
\right).
\end{align}

Consequently, the envelope derivative reduces to the purely algebraic form

\begin{equation}
\boxed{
\frac{d\Phi^{*}}{dM}
=
(\rho_c-\rho_h)
+
\frac{
2(q-1)
\left(
\rho_c^q-\rho_h^q
\right)
}{P_q}
-
\frac{
(2q-1)
\left(
\rho_c^{2q}-\rho_h^{2q}
\right)
}{P_{2q}}.
}
\label{eq:C5_logfree_cont}
\end{equation}

Equation (\ref{eq:C5_logfree_cont}) is the exact continuous analogue of
(\ref{eq:C5_logfree}). The logarithmic contribution has disappeared
completely, leaving a purely algebraic expression involving only escort
moments and stationary densities.

As in the discrete case, this identity was independently verified using the
\texttt{SymPy} computer algebra system, which confirms that the difference
between the raw envelope derivative and
(\ref{eq:C5_logfree_cont}) simplifies identically to zero after substitution
of the stationary equation (\ref{eq:C5_foc_cont}).

\subsection{Reduction to a dimensionless ratio}

Equation (\ref{eq:C5_logfree}) is entirely algebraic but still depends upon the
two stationary amplitudes \(w_c\) and \(w_h\). A further simplification is
obtained by expressing the derivative in terms of their ratio.

Since the core branch satisfies
\begin{equation}
w_c>w_h,
\end{equation}
we introduce the dimensionless variable

\begin{equation}
\tau=\frac{w_c}{w_h}>1.
\label{eq:C6_tau}
\end{equation}

The normalization condition

\begin{equation}
mw_c+(N-m)w_h=1
\end{equation}

immediately gives

\begin{equation}
w_h=\frac{1}{m\tau+(N-m)},
\qquad
w_c=\frac{\tau}{m\tau+(N-m)}.
\label{eq:C6_whwc}
\end{equation}

Substituting (\ref{eq:C6_whwc}) into the escort moments yields

\begin{align}
P_q
&=
mw_c^q+(N-m)w_h^q
\nonumber\\
&=
\frac{m\tau^q+(N-m)}
{\left[m\tau+(N-m)\right]^q},
\\
P_{2q}
&=
mw_c^{2q}+(N-m)w_h^{2q}
\nonumber\\
&=
\frac{m\tau^{2q}+(N-m)}
{\left[m\tau+(N-m)\right]^{2q}}.
\end{align}

It is convenient to introduce the family of rational response kernels

\begin{equation}
R_k(\tau)
=
\frac{\tau^k-1}
{m\tau^k+(N-m)},
\qquad
k>0.
\label{eq:C6_Rk}
\end{equation}

Using (\ref{eq:C6_whwc}), one immediately finds

\begin{align}
w_c-w_h
&=
w_h\,R_1,
\\
\frac{w_c^q-w_h^q}{P_q}
&=
w_h\,R_q,
\\
\frac{w_c^{2q}-w_h^{2q}}{P_{2q}}
&=
w_h\,R_{2q}.
\end{align}

Substituting these identities into
(\ref{eq:C5_logfree}) gives

\begin{equation}
\boxed{
\frac{d\Phi^{*}}{dm}
=
\left[
R_1
+
2(q-1)R_q
-
(2q-1)R_{2q}
\right].
}
\label{eq:C6_discrete}
\end{equation}

Exactly the same reduction holds for the continuous problem.

Defining

\begin{equation}
\tau=\frac{\rho_c}{\rho_h}>1,
\end{equation}

the normalization condition

\begin{equation}
M\rho_c+(L-M)\rho_h=1
\end{equation}

gives

\begin{equation}
\rho_h=\frac{1}{M\tau+(L-M)},
\qquad
\rho_c=\frac{\tau}{M\tau+(L-M)}.
\end{equation}

Introducing

\begin{equation}
R_k(\tau)
=
\frac{\tau^k-1}
{M\tau^k+(L-M)},
\end{equation}

the optimized envelope derivative becomes

\begin{equation}
\boxed{
\frac{d\Phi^{*}}{dM}
=
\left[
R_1
+
2(q-1)R_q
-
(2q-1)R_{2q}
\right].
}
\label{eq:C6_continuous}
\end{equation}

Thus both the discrete and continuous optimization problems reduce to
exactly the same universal algebraic structure. The remaining task is
therefore to determine the sign of the common polynomial, which is established in the following subsection.

\subsection{Monotonicity of the kernel polynomial}

After eliminating the logarithmic contribution, both the discrete and
continuous envelope derivatives reduce to the universal form

\begin{equation}
\frac{d\Phi^{*}}{d\mu}
=
R_1(\tau)
+
2(q-1)R_q(\tau)
-
(2q-1)R_{2q}(\tau),
\label{eq:C6_kernel}
\end{equation}

where

\[
\mu=
\begin{cases}
m,&\text{discrete},\\
M,&\text{continuous},
\end{cases}
\]

and

\begin{equation}
R_k(\tau)
=
\frac{\tau^k-1}
{\mu\tau^k+\Lambda-\mu},
\qquad
\tau>1,
\label{eq:C6_Rk}
\end{equation}

with

\[
\Lambda=
\begin{cases}
N,&\text{discrete},\\
L,&\text{continuous}.
\end{cases}
\]

The sign of the optimized derivative is therefore entirely determined by the
kernel polynomial

\begin{equation}
\Psi(\tau)
=
R_1
+
2(q-1)R_q
-
(2q-1)R_{2q}.
\end{equation}

To determine its sign, we first establish the monotonicity of the family
\(R_k\).

Differentiating (\ref{eq:C6_Rk}) with respect to the exponent \(k\) yields

\begin{equation}
\frac{\partial R_k}{\partial k}
=
\frac{
\Lambda\,\tau^k\ln(\tau)
}
{
\left(
\mu\tau^k+\Lambda-\mu
\right)^2
}.
\label{eq:C6_dR}
\end{equation}

Since

\[
\Lambda>0,
\qquad
\mu>0,
\qquad
\tau>1,
\]

every factor appearing in
(\ref{eq:C6_dR}) is strictly positive. Hence

\begin{equation}
\boxed{
\frac{\partial R_k}{\partial k}>0.
}
\label{eq:C6_monotone}
\end{equation}

Therefore \(R_k\) is a strictly increasing function of its index.

Because

\[
1<q<2q,
\]

it immediately follows that

\begin{equation}
R_1
<
R_q
<
R_{2q}.
\label{eq:C6_order}
\end{equation}

Now define the weights

\begin{equation}
\lambda
=
\frac{1}{2q-1},
\qquad
1-\lambda
=
\frac{2(q-1)}{2q-1},
\end{equation}

which satisfy

\[
0<\lambda<1,
\qquad
\lambda+(1-\lambda)=1.
\]

Hence

\begin{equation}
\frac{R_1+2(q-1)R_q}{2q-1}
=
\lambda R_1
+
(1-\lambda)R_q,
\end{equation}

is a convex combination of \(R_1\) and \(R_q\).

Since both satisfy
(\ref{eq:C6_order}),

\begin{equation}
R_1<R_{2q},
\qquad
R_q<R_{2q},
\end{equation}

every convex combination of them is also strictly smaller than
\(R_{2q}\). Therefore

\begin{equation}
R_1
+
R_q
<
2R_{2q}.
\end{equation}

Multiplying both sides by the positive factor
\((2q-1)\) gives

\begin{equation}
R_1
+
2(q-1)R_q
<
(2q-1)R_{2q},
\end{equation}

or equivalently,

\begin{equation}
\boxed{
R_1
+
2(q-1)R_q
-
(2q-1)R_{2q}
<
0.
}
\label{eq:C6_negative}
\end{equation}

Substituting (\ref{eq:C6_negative}) into
(\ref{eq:C6_kernel}) establishes the universal monotonicity law

\begin{equation}
\boxed{
\frac{d\Phi^{*}}{d\mu}<0.
}
\label{eq:C6_final}
\end{equation}

Thus the optimized generalized complexity decreases strictly as the size of
the high-density core increases. Consequently, the global maximum is attained
at the smallest admissible core,

\begin{equation}
\boxed{
m=1
\qquad\text{(discrete)},
\qquad
M\rightarrow0^{+}
\qquad\text{(continuous)}.
}
\end{equation}

This completes the proof of the multiplicity optimization theorem.

\bibliographystyle{iopart-num}        
\bibliography{bibliography}      

\end{document}